\documentclass[12pt,reqno,a4paper]{amsart}
\usepackage{amsmath,amssymb,dsfont,graphicx,cite,latexsym,
epsf,cancel,epic,eepic}
\usepackage{amssymb,mathrsfs}
\usepackage[colorlinks=false]{hyperref}
\usepackage{mathpazo}
\usepackage{color}
\newtheorem{theorem}{Theorem}
\newtheorem{proposition}[theorem]{Proposition}

\newtheorem{lemma}[theorem]{Lemma}

\newtheorem{definition}[theorem]{Definition}
\newtheorem{cor}[theorem]{Corollary}
\def\beq{\begin{equation}}
\def\eeq{\end{equation}}
\def\bea{\begin{eqnarray}}
\def\eea{\end{eqnarray}}
\def\benpf{\noindent {\textbf{{\emph{Proof.}}\;}}}
\def\endpf{\hfill$\blacksquare$\medskip}
\makeatletter
\expandafter\let\expandafter
\reset@font\csname reset@font\endcsname
\def\subeqnarray{\arraycolsep1pt
   \def\@eqnnum\stepcounter##1{\stepcounter{subequation}
       {\reset@font\rm(\theequation\alph{subequation})}}
\jot5mm     \eqnarray}

\makeatother
\newcommand{\bbR}{{\mathbb R}}

\newcommand{\cA}{{\mathcal A}}
\newcommand{\cB}{{\mathcal B}}

\newcommand{\cH}{{\mathcal H}}

\newcommand{\cP}{{\mathcal P}}

\def\ep{\varepsilon}
\def\epsilon{\varepsilon}
\def\t{\widetilde}

\def\tilde{\widetilde}

\def\endpf{\hfill$\square$\medskip}

\newbox\meibox
\def\placeunder#1#2#3#4{\setbox\meibox%
\vbox{\hbox{\hskip#4$\hphantom{#2}$}\hbox{$\hphantom{#1}$}}%
\vtop{\baselineskip=0pt\lineskiplimit=\baselineskip%
\lineskip=#3\hbox to \wd\meibox{\hfil\hskip#4$#2$\hfil}%
\hbox to \wd\meibox{\hfil$#1$\hfil}}}
\def\undertilde#1{\mathchoice{%
\placeunder{\vbox to 1.4pt{\hbox{$\displaystyle\widetilde{\,\,\,
}$}\vss}}{\displaystyle#1}{1.5pt}{1.5pt}}%
{\placeunder{\vbox to 1.4pt{\hbox{$\textstyle\widetilde{\,\,
}$}\vss}}{\textstyle#1}{1.5pt}{1.5pt}}%
{\placeunder{\vbox to 1.4pt{\hbox{$\scriptstyle\tilde{
}$}\vss}}{\scriptstyle#1}{1pt}{1pt}}%
{\placeunder{\vbox to 1.4pt{\hbox{$\scriptscriptstyle\tilde{
}$}\vss}}{\scriptscriptstyle#1}{1pt}{1pt}}%
}
\def\intprod{\mathbin{\hbox to 6pt{%
                 \vrule height0.4pt width5pt depth0pt
                 \kern-.4pt
                 \vrule height6pt width0.4pt depth0pt\hss}}}

\begin{document}
\title[Commuting systems of integrable symplectic birational maps]
{A construction of commuting systems \\ of integrable symplectic birational maps.\\ Lie-Poisson case}

\author{Matteo Petrera \and Yuri B. Suris }

\date{\today}

\thanks{E-mail: {\tt  petrera@math.tu-berlin.de, suris@math.tu-berlin.de}}

\maketitle

\begin{center}
{\footnotesize{
Institut f\"ur Mathematik, MA 7-1\\
Technische Universit\"at Berlin, Str. des 17. Juni 136,
10623 Berlin, Germany
}}

\end{center}


\begin{abstract}
We give a construction of completely integrable ($2n$)-dimensional Hamiltonian systems with symplectic brackets of the Lie-Poisson type (linear in coordinates) and with quadratic Hamilton functions. Applying to any such system the so called Kahan-Hirota-Kimura discretization scheme, we arrive at a birational ($2n$)-dimensional map. We show that this map is symplectic with respect to a symplectic structure that is a perturbation of the original symplectic structure on $\mathbb R^{2n}$, and possesses $n$ independent integrals of motion, which are perturbations of the original Hamilton functions and are in involution with respect to the invariant symplectic structure. Thus, this map is completely integrable in the Liouville-Arnold sense. Moreover, under a suitable normalization of the original $n$-tuples of commuting vector fields, their Kahan-Hirota-Kimura discretizations also commute and share the invariant symplectic structure and the $n$ integrals of motion. This paper extends our previous ones, {\tt arXiv:1606.08238 [nlin.SI]} and {\tt arXiv:1607.07085 [nlin.SI]}, where similar results were obtained for Hamiltonian systems with a constant (canonical) symplectic structure and cubic Hamilton functions.
\end{abstract}

\section{Introduction}
\label{sect intro}

In the recent papers \cite{PS4dim, PS6dim}, we introduced a large family of integrable symplectic maps, appearing as the so called Kahan-Hirota-Kimura discretization of a big family of completely integrable Hamiltonian systems in arbitrary even dimension $2n$, with the canonical  symplectic structure and cubic Hamilton functions. We mentioned there an open problem of generalizing these findings for the case of Hamiltonian systems with a Poisson tensor linear in local coordinates (that is, a Lie-Poisson tensor) and quadratic Hamilton functions. Here, such a generalization is achieved.

We consider a certain family of Lie-Poisson tensors $J(x)$ of full rank (thus, defining symplectic structures) on $\mathbb R^{2n}$. For such a  tensor, there exist constant $2n \times 2n$ matrices $A$ satisfying  
\begin{equation} \label{skew}
A^{\rm T}J(x)=J(x)A, \quad \forall x\in\mathbb R^{2n}.
\end{equation}
Clearly, any power of $A$ satisfies the same equation. Generically, along with $A$, one has an $n$-dimensional vector space of matrices satisfying \eqref{skew} which consists of polynomials of $A$ of degree $n-1$. To each non-degenerate matrix $A$ with property \eqref{skew}, there corresponds a vector space of quadratic polynomials $H_0(x)$ on $\mathbb R^{2n}$, satisfying a system of second order linear PDEs encoded in the matrix equation 
\beq \label{harm}
A(\nabla^2 H)=(\nabla^2 H) A^{\rm T},
\end{equation}
where $\nabla^2 H$ is the Hesse matrix of the function $H$. Such a polynomial  $H_0(x)$ can be included to an $n$-tuple of quadratic polynomials $H_i(x)$ satisfying the same matrix differential equations \eqref{harm}, and characterized by 
\begin{equation}\label{CR}\nonumber
\nabla H_i(x)=A\nabla H_{i-1}(x), \quad i=1,\ldots, n-1.
\end{equation}
For a generic $A$, the functions $H_i(x)$, $i=0,\ldots, n-1$, are functionally independent and are in involution with respect to the Lie-Poisson structure on $\mathbb R^{2n}$ defined by the tensor $J(x)$. Thus, the flows of the Hamiltonian vector fields $f_i(x)=J(x)\nabla H_i(x)$ commute, and comprise a completely integrable Hamiltonian system. 

When applied to a completely integrable Hamiltonian system $\dot{x}=f_0(x)$ of this family, the Kahan-Hirota-Kimura discretization method produces the map $\Phi_{f_0}$ with the following striking properties.
\begin{itemize}
\item The map $\Phi_{f_0}$ is Poisson with respect to a (symplectic) Poisson structure on $\mathbb R^{2n}$ which is a perturbation of the original one (defined by $J(x)$), and possesses $n$ functionally independent integrals in involution. In other words, the map $\Phi_{f_0}$ is completely integrable. 

\item In general, the maps $\Phi_{f_i}$ do not commute among themselves. However, one can find systems of commuting  maps which include 
$\Phi_{f_0}$. We say that a linear combination of the vector fields,
$$
\sum_{i=0}^{n-1} \alpha_i f_i(x) = J(x)\left(\sum_{i=0}^{n-1} \alpha_iA^i\right)\nabla H_0(x)= J(x)B\nabla H_0(x)= B^{\rm T}f_0(x),
$$
is {\em associated} to the vector field $f_0(x)$, if the matrix $B=\sum_{i=0}^{n-1} \alpha_iA^i$  satisfies $B^2=I$. Equivalently, the polynomial $B(\lambda)=\sum_{i=0}^{n-1} \alpha_i\lambda^i$ sends each of the $n$ distinct eigenvalues of $A$ to $\pm 1$. This defines an equivalence relation on the set of vector fields $J(x)\nabla H(x)$ with $H(x)$ satisfying \eqref{harm}. It turns out that Kahan-Hirota-Kimura discretizations of associated vector fields commute and share the invariant symplectic structure and $n$ functionally independent integrals. 

\item The common integrals $\t H(x,\epsilon)$ of the Kahan-Hirota-Kimura discretizations of the associated vector fields are rational perturbations of the original polynomial Hamilton functions $$H(x)=\sum_{i=0}^{n-1}\alpha_i H_i(x),$$ and satisfy the same second order differential equation \eqref{harm} as $H(x)$ do.

\end{itemize}

We give a quick review of the Kahan-Hirota-Kimura discretization method in Sect. \ref{sect HK}. The specialization of this method for Lie-Poisson systems is considered in Sect. \ref{sect KHK for LP}. Then, in Sect. \ref{sect int fam}, we discuss details of the general construction of completely integrable Hamiltonian systems generated by a Lie-Poisson tensor $J(x)$ and a  matrix $A$ related as in \eqref{skew}. Algebraic properties of the corresponding vector fields are collected in Sect. \ref{sect vector fields}. Associated vector fields are introduced in Sect. \ref{sect assoc}. We prove the main results in Sect. \ref{sect commut} (commutativity), Sect. \ref{sect integrals} (integrals of motion), Sect. \ref{sect symplectic} (invariant symplectic structure) and Sect. \ref{sect diff eqs} (differential equations for the integrals of the maps).

\section{General properties of the Kahan-Hirota-Kimura discretization}
\label{sect HK}

Here we recall the main facts about the Kahan-Hirota-Kimura discretization.

This method was introduced in the geometric integration literature by Kahan in the unpublished notes \cite{K} as a method applicable to any system of ordinary differential equations on $\bbR^N$ with a quadratic vector field:
\begin{equation}\label{eq: diff eq gen} \nonumber
\dot{x}=f(x)=Q(x)+Bx+c,
\end{equation}
where each component of $Q:\bbR^N\to\bbR^N$ is a quadratic form, while $B\in{\rm Mat}_{N\times N}(\bbR)$ and $c\in\bbR^N$. Kahan's discretization (with stepsize $2\epsilon$) reads as
\begin{equation}\label{eq: Kahan gen}
\frac{\widetilde{x}-x}{2\epsilon}=2f\Big(\frac{x+\t x}{2}\Big)-\frac{1}{2}f(x)-\frac{1}{2}f(\t x)=Q(x,\widetilde{x})+\frac{1}{2}B(x+\widetilde{x})+c,
\end{equation}
where
\[
Q(x,\widetilde{x})=\frac{1}{2}\big(Q(x+\widetilde{x})-Q(x)-Q(\widetilde{x})\big)
\]
is the symmetric bilinear form corresponding to the quadratic form $Q$. Equation (\ref{eq: Kahan gen}) is {\em linear} with respect to $\widetilde x$ and therefore defines a {\em rational} map $\widetilde{x}=\Phi_f(x,\epsilon)$. Clearly, this map approximates the time $2\epsilon$ shift along the solutions of the original differential system. Since equation (\ref{eq: Kahan gen}) remains invariant under the interchange $x\leftrightarrow\widetilde{x}$ with the simultaneous sign inversion $\epsilon\mapsto-\epsilon$, one has the {\em reversibility} property
\begin{equation}\label{eq: reversible}
\Phi_f^{-1}(x,\epsilon)=\Phi_f(x,-\epsilon).
\end{equation}
In particular, the map $f$ is {\em birational}. The explicit form of the map $\Phi_f$ defined by \eqref{eq: Kahan gen} is 
\beq \label{eq: Phi gen}
\t x =\Phi_f(x,\ep)= x + 2\ep \left( I - \ep f'(x) \right)^{-1} f(x),
\eeq
where $f'(x)$ denotes the Jacobi matrix of $f(x)$. Moreover, if the vector field $f(x)$ is homogeneous (of degree 2), then \eqref{eq: Phi gen} can be equivalently rewritten as
\beq \label{eq: Phi hom}
\t x =\Phi_f(x,\ep)= \left( I - \epsilon f'(x) \right)^{-1} x.
\eeq
Due to \eqref{eq: reversible}, in the latter case we also have:
\beq \label{eq: Phi hom alt}
x =\Phi_f(\t x,-\ep)= \left( I + \epsilon f'(\t x) \right)^{-1} \t x \quad \Leftrightarrow \quad \t x = \left( I + \epsilon f'(\t x) \right) x.
\eeq
One has the following expression for the Jacobi matrix of the map $\Phi_f$:
\beq \label{Jac}
d\Phi_f(x)=\frac{\partial\t x}{\partial x}=\big(I-\epsilon f'(x)\big)^{-1}\big(I+\epsilon f'(\t x)\big).
\eeq

Kahan applied this discretization scheme to the famous Lotka-Volterra system and showed that in this case it possesses a very remarkable non-spiralling property. This property was explained by Sanz-Serna \cite{SS} by demonstrating that in this case the numerical method preserves an invariant Poisson structure of the original system.

The next intriguing appearance of this discretization was in two papers by Hirota and Kimura who (being apparently unaware of the work by Kahan) applied it to two famous {\em integrable} system of classical mechanics, the Euler top and the Lagrange top \cite{HK, KH}. Surprisingly, the discretization scheme produced in both cases {\em integrable} maps. 

In \cite{PS, PPS1, PPS2} the authors undertook an extensive study of the properties of the Kahan's method when applied to integrable systems (we proposed to use in the integrable context the term ``Hirota-Kimura method''). It was demonstrated that, in an amazing number of cases, the method preserves integrability in the sense that the map $\Phi_f(x,\epsilon)$ possesses as many independent integrals of motion as the original system $\dot x=f(x)$.

Further remarkable geometric properties of the Kahan's method were discovered by Celledoni, McLachlan, Owren and Quispel in \cite{CMOQ1, CMOQ2}. They demonstrated that for an arbitrary Hamiltonian vector field $f(x)=J\nabla H(x)$ with a constant Poisson tensor $J$ and a cubic Hamilton function $H(x)$, the map $\Phi_f(x,\epsilon)$ possesses a rational integral of motion 
as well as an invariant measure with a rational density. These properties are unrelated to integrability. 

Finally, in papers  \cite{PS4dim, PS6dim} which are direct precursors of the present one, there were discovered commuting families of completely integrable Hirota-Kimura maps, along with their full sets of integrals and with an invariant symplectic structure.


\section{General properties of the Kahan-Hirota-Kimura discretization \\ applied to Lie-Poisson type systems}
\label{sect KHK for LP}

In this paper, we consider the following class of vector fields on $\mathbb R^N$:
\beq \label{vf LP}
f(x)=J(x)\nabla H(x),
\eeq
where $J(x)$ is a $N\times N$ matrix whose entries are linear forms in $x$, and $H(x)$ is a quadratic form in $x$. We have:
\beq \label{f'} 
f'(x)=J(x) \nabla^2 H + J'\nabla H(x).
\eeq
Here $\nabla^2 H$ is the (constant) Hesse matrix of $H(x)$, $J'$ is a (constant) tensor such that, for any $y\in\bbR^N$, $J'y$ is the $N\times N$ matrix with the entries 
$$
(J'y)_{ik}=\sum_{k=1}^N \frac{\partial J_{ij}(x)}{\partial x_k}y_j .
$$
From this we derive the following identity which we mention for the later reference:
\beq \label{J'}
(J'y)z=J(z)y, \quad y,z\in\bbR^N.
\eeq

\begin{proposition}
For any vector field of the form \eqref{vf LP}, the Kahan map $\t x=\Phi_f(x)$ can be implicitly written as
\beq \label{Phi x-x}
\t x-x=2\epsilon \Big(I-\epsilon J(x)\nabla^2 H \Big)^{-1} J\!\left( \frac{x+\t x}{2}\right) \nabla H(x),
\eeq
or, alternatively, as
\beq \label{Phi x}
\t x = \Big(I-\epsilon J(x)\nabla^2 H \Big)^{-1}\Big(I+\epsilon J(\t x)\nabla^2 H \Big)x.
\eeq
\end{proposition}
\begin{proof}
From \eqref{f'} we derive:
$$
I-\epsilon f'(x)=\Big(I-\epsilon J(x)\nabla^2 H\Big)\Big(I-\epsilon\big(I-\epsilon J(x)\nabla^2 H\big)^{-1} J'\nabla H(x)\Big).
$$
Therefore, equation \eqref{eq: Phi gen} can be rewritten as
$$
\t x -x=2\epsilon\Big(I-\epsilon\big(I-\epsilon J(x)\nabla^2 H\big)^{-1} J'\nabla H(x)\Big)^{-1}\Big(I-\epsilon J(x)\nabla^2 H\Big)^{-1}f(x).
$$
Multiplying this equation with the inverse of the first factor on the right-hand side, and taking into account equations \eqref{J'} and \eqref{vf LP}, we find:
$$
\t x -x-\epsilon \Big(I-\epsilon J(x)\nabla^2 H\Big)^{-1} J(\t x-x)\nabla H(x)=2\epsilon \Big(I-\epsilon J(x)\nabla^2 H\Big)^{-1} J(x)\nabla H(x).
$$
There follows, due to linearity of $J(x)$:
$$
\t x -x=\epsilon \Big(I-\epsilon J(x)\nabla^2 H\Big)^{-1} J(\t x+x)\nabla H(x),
$$
which coincides with \eqref{Phi x-x}.

From the latter equation we further derive, taking into account that $\nabla H(x)=(\nabla^2 H) x$:
\begin{eqnarray*}
\t x & = & \Big(I-\epsilon J(x)\nabla^2 H\Big)^{-1}\Big(I-\epsilon J(x)\nabla^2 H+\epsilon J(\t x+x)\nabla^2 H\Big)x\\
      & = & \Big(I-\epsilon J(x)\nabla^2 H\Big)^{-1}\Big(I+\epsilon J(\t x)\nabla^2 H\Big)x,
\end{eqnarray*}
which finishes the proof. 
\end{proof}

\section{A family of integrable Lie-Poisson  systems}
\label{sect int fam}


Starting from this point, we always set $N=2n$. Thus, we consider the  phase space $\bbR^{2n}$ with coordinates $x=(x_1,\ldots,x_{2n})^{\rm T}$. We write
\beq
x=\begin{pmatrix} u \\ v \end{pmatrix}, \quad u=(x_1,\ldots,x_{n})^{\rm T}, \quad v=(x_{n+1},\ldots,x_{2n})^{\rm T}.
\eeq 
Equip the phase space with the Lie-Poisson tensor
\beq \label{P full}
J(x)  =  \begin{pmatrix} 0 & X(u) \\ -X(u) & 0 \end{pmatrix},
\eeq
where
\beq \label{X full}
X(u)=X^{\rm T}(u)=\begin{pmatrix}
x_1 & x_2 & \cdots & x_n\\
x_2 & x_3 & \cdots & x_1\\
 \cdots & \cdots & \cdots & \cdots \\
 x_n & x_1 & \cdots & x_{n-1} 
\end{pmatrix}  
\eeq
is an $n\times n$ cyclic Hankel (therefore symmetric) matrix.  The rank of $J(x)$ at a generic point  $x\in\bbR^{2n}$ is equal to $2n$. To check that $J(x)$ is a Poisson tensor, one has to verify the Jacobi identity
\beq \label{Jacobi id}
\pi_{ijk}=\{x_i,\{x_j,x_k\}\}+\{x_j,\{x_k,x_i\}\}+\{x_k,\{x_i,x_j\}\}=0
\eeq
for all possible triples of indices $\{i,j,k\}$ from $\{1,2,\ldots,2n\}$. Due to the block-diagonal structure of the matrix $J(x)$, one only has to check this for the cases $i,j\in\{1,\ldots,n\}$, $k\in\{n+1,\ldots,2n\}$ and $i,j\in\{n+1,\ldots,2n\}$, $k\in\{1,\ldots,n\}$. Moreover, in these cases \eqref{Jacobi id} simplifies to
\beq \label{Jacobi id block}
\pi_{ijk}=\{x_i,\{x_j,x_k\}\}-\{x_j,\{x_i,x_k\}\}=0.
\eeq
Both terms on the right-hand side vanish if $i,j\in\{1,\ldots,n\}$, $k\in\{n+1,\ldots,2n\}$. In the remaining case $i,j\in\{n+1,\ldots,2n\}$, $k\in\{1,\ldots,n\}$, we compute: 
$$
\{x_j,x_k\}= -x_{j+k-1\, ({\rm mod}\, n)}, 
$$
and further
$$
\{x_i,\{x_j,x_k\}\}  =  -\{x_i, x_{j+k-1\, ({\rm mod}\, n)}\}
  =  x_{i+j+k-2\, ({\rm mod}\, n)}.
$$
Since this expression is symmetric with respect to $i\leftrightarrow j$, we see that \eqref{Jacobi id block} is satisfied.
\medskip

For any function $H_0(x)$ on  $\bbR^{2n}$, the corresponding Hamiltonian system is governed by the equations of motion
\beq  \label{Ham} \nonumber
\dot x= J(x) \nabla H_0(x).
\eeq
\begin{proposition}
Consider a constant non-degenerate $2n\times 2n$ matrix $A$, and suppose that for the functions $H_1(x), \ldots, H_{n-1}(x)$ the following relations are satisfied:
\beq \label{grad K}
\nabla H_i(x) = A \nabla H_{i-1}(x), \quad i=1,\ldots,n-1.
\eeq
If the matrix $A$ satisfies
\beq \label{cond A}
 A^{\rm{T}} J(x) = J(x) A
\eeq
for any $x\in \bbR^{2n}$,
then the functions $H_i(x)$ are pairwise in involution.
\end{proposition}

\benpf Let $0\le i<j\le n-1$. Then $\nabla H_j=A^{j-i}\nabla H_i$. We have:
$$
\{H_i,H_j\} = (\nabla H_i)^{\rm{T}} J(x) \nabla H_j=(\nabla H_i)^{\rm{T}} J(x) A^{j-i} \nabla H_i.
$$
Since for any $\ell\in\mathbb N$ the matrix $A^\ell$ satisfies the same condition as \eqref{cond A}, that is,
$$
(A^\ell)^{\rm{T}} J(x) = J(x) A^\ell,
$$
that is, $J(x)A^\ell$ is skew-symmetric, we conclude that $\{H_i,H_j\}=0$. \endpf

If the minimal annulating polynomial of the matrix $A$ has degree $n$, then the matrices $I,A,\ldots,A^{n-1}$  are linearly independent, and then equation \eqref{grad K} ensures that $H_0,\ldots,H_{n-1}$ are generically functionally independent. Effectively, we are considering a family of functions $H(x)$ such that
\beq \label{grad K gen}
\nabla H(x)=(\beta_0 I+\beta_1 A+\ldots+\beta_{n-1} A^{n-1} )\nabla H_0(x).
\eeq

We now discuss applicability of this construction. For a given function $H_0$, differential equations \eqref{grad K} for $H_1$ are solvable if and only if $H_0$ satisfies the following condition:
\beq  \label{d2H}
A(\nabla^2 H_0)=(\nabla^2 H_0) A^{\rm T},
\eeq
where $\nabla^2 H_0$ is the Hesse matrix of the function $H_0$. If this condition is satisfied, then for solutions of \eqref{grad K} we find:
$$
A(\nabla^2 H_1)=A^2(\nabla^2 H_0)=A(\nabla^2 H_0) A^{\rm T}=(\nabla^2 H_1) A^{\rm T}.
$$
Thus, $H_1$ satisfies the same condition \eqref{d2H}. By induction, the same is true for all $H_i$, thus for all functions $H$ satisfying \eqref{grad K gen}. 

\begin{proposition}
For the matrix $J(x)$ given in \eqref{P full}, \eqref{X full}, the set of matrices $A$ satisfying \eqref{cond A} is given by 
\beq \label{A gen}
A = \begin{pmatrix}
{\mathcal A} & 0 \\
0 & {\mathcal A}
\end{pmatrix}=\mathcal A \oplus \mathcal A,
\eeq
where ${\mathcal A}$ is a circulant $n\times n$ matrix. 
\end{proposition}
\begin{proof} We write
\beq  \label{P short}
J(x)= x_1J_1+x_2J_2+\ldots+x_nJ_n, \quad {\rm where}\quad J_k=\begin{pmatrix} 0 & Q_k \\ -Q_k & 0 \end{pmatrix},
\eeq
and $Q_k$ are $n\times n$ cyclic Hankel matrices with the entries 
\beq 
(Q_k)_{i,j}=\left\{\begin{array}{ll} 1, &  i+j=k+1\ ({\rm mod}\ n), \\ 0, & {\rm else}. \end{array}\right.
\eeq
Condition \eqref{cond A} is equivalent to
\beq
A^{\rm{T}} J_k = J_k A, \quad k=1,\ldots,n.
\eeq
If $A$ is written in the block form as
\beq \label{A block prelim} \nonumber
A = \begin{pmatrix}
{\mathcal A}_1 & {\mathcal A}_2 \\
{\mathcal A}_3 & {\mathcal A}_4
\end{pmatrix},
\eeq
with $n\times n$ blocks ${\mathcal A}_i$, then condition \eqref{cond A} reads:
\beq \label{cond2}
{\mathcal A}_1^{\rm T}Q_{k}=Q_{k}{\mathcal A}_4 , 
\eeq
\beq \label{cond3}
{\mathcal A}_2^{\rm T}Q_{k}+Q_{k}{\mathcal A}_2=0, \quad {\mathcal A}_3^{\rm T}Q_{k}+Q_{k}{\mathcal A}_3=0
\eeq
for all $k=1,\ldots, n$. One easily shows that:
\begin{itemize}
\item Conditions \eqref{cond3} force the matrices ${\mathcal A}_2$, ${\mathcal A}_3$ to vanish. Indeed, we find:
$$
({\mathcal A}_2)_{i,j}=-({\mathcal A}_2)_{k+1-j,k+1-i} \quad {\rm for\quad all}\quad k=1,\ldots, n.
$$
For $k=i+j-1$ we arrive at $({\mathcal A}_2)_{i,j}=0$.
\item Conditions \eqref{cond2} force that the matrix ${\mathcal A}_1={\mathcal A}_4$ is circulant. Indeed, we find:
$$
({\mathcal A}_1)_{i,j}=({\mathcal A}_4)_{k+1-j,k+1-i} \quad {\rm for\quad all}\quad k=1,\ldots, n.
$$
For $k=i+j-1$ this gives $({\mathcal A}_1)_{i,j}=({\mathcal A}_4)_{i,j}$. Moreover, for $k=i+j+\ell-1$, we arrive at $({\mathcal A}_1)_{i,j}=({\mathcal A}_4)_{i+\ell,j+\ell}=({\mathcal A}_1)_{i+\ell,j+\ell}$. Thus, shifting both indices $i,j$ by the same amount does not change the value of $({\mathcal A}_1)_{i,j}$.  \end{itemize}
We finally arrive at \eqref{A gen}.
\end{proof}
For any such matrix ${\mathcal A}$, the set of all linear combinations $\beta_0 I+\beta_1 {\mathcal A}+\ldots+\beta_{n-1}{\mathcal A}^{n-1}$ is a vector subspace of the space of circulant matrices. Thus, without restricting the generality, we can set from the beginning:
\beq \label{A final}
A =\begin{pmatrix} \cP & 0 \\ 0 & \cP \end{pmatrix} = \cP\oplus \cP, \qquad \cP=\begin{pmatrix} 0 & 1 & 0 & \cdots & 0 & 0 \\
                              0 & 0 & 1 & \cdots & 0 & 0 \\
                              \cdots & \cdots  & \cdots  & \cdots  & \cdots  & \cdots \\
                              0 & 0 & 0 & \cdots & 0 & 1 \\
                              1 & 0 & 0 & \cdots & 0 & 0
                              \end{pmatrix}.
\eeq
$\cP$ is the $n\times n$ cyclic shift matrix, satisfying $\cP^{\rm T}=\cP^{-1}=\cP^{n-1}$.
\smallskip

With this matrix $A$, the quadratic solutions of the matrix differential equation  \eqref{d2H} are easily characterized.
\begin{proposition}
A quadratic form 
\beq\label{H quad} \nonumber
H_0(x)=\frac{1}{2}\sum_{i,j=1}^{2n} h_{ij}x_ix_j
\eeq
with a constant symmetric matrix $\nabla^2 H_0=(h_{ij})_{i,j=1}^{2n}$ satisfies matrix differential equation \eqref{d2H} with the matrix $A$ from \eqref{A final} if and only if 
\beq \label{cH}
\nabla^2 H_0=\begin{pmatrix} \mathcal H_1 & \mathcal H_2 \\ \mathcal H_2 & \mathcal H_3 \end{pmatrix},
\eeq
where $n\times n$ blocks $\mathcal H_1,\mathcal H_2, \mathcal H_3$ are cyclic Hankel (therefore symmetric) matrices. The dimension of the space of solutions of \eqref{d2H} is $3n$.
\end{proposition}
\begin{proof}
Clearly, matrix equation \eqref{d2H} is equivalent to $\cP\mathcal H_k=\mathcal H_k \cP^{\rm T}$ for $k=1,2,3$. This is equivalent to $(\mathcal H_k)_{i,j}=(\mathcal H_k)_{i-1,j+1}$, so that $(\mathcal H_k)_{i,j}$ only depends on $i+j$, that is, $\mathcal H_k$ is cyclic Hankel.
\end{proof}

The two classes of $n\times n$ matrices: {\em circulant} ones, characterized by $\cP\mathcal A=\mathcal A\cP$, and {\em cyclic Hankel} ones, characterized by $\cP\mathcal A=\cA\cP^{-1}$, will play a fundamental role in our work. For further reference, we give an important property of the latter class.
\begin{lemma}\label{lemma ABC}
If $\mathcal A$, $\mathcal B$, $\mathcal C$ are three cyclic Hankel matrices, then
\beq \label{ABC}
\mathcal A\mathcal B\mathcal C=\mathcal C\mathcal B\mathcal A,
\eeq
and this product is a cyclic Hankel matrix.
\end{lemma}
\begin{proof}
Since all three matrices $\mathcal A$, $\mathcal B$, $\mathcal C$ are symmetric, equation \eqref{ABC} is equivalent to the statement that the matrix $\mathcal A\mathcal B\mathcal C$ is symmetric. But this is a consequence of the fact that this matrix is cyclic Hankel, which follows from 
$$
\cP(\mathcal A\mathcal B\mathcal C)\cP=\cP\mathcal A \cP\cdot \cP^{-1}{\mathcal B}\cP^{-1}\cdot \cP\mathcal C \cP=\mathcal A\mathcal B\mathcal C.
\qquad\qquad\qedhere
$$
\end{proof}
Note that \eqref{ABC} is easily generalized for a product of any odd number of cyclic Hankel matrices.

\section{General algebraic properties of the vector fields $f_i(x)=J(x)\nabla H_i(x)$}
\label{sect vector fields}

Assume that $H_0(x)$ is a homogeneous quadratic polynomial satisfying \eqref{d2H}.
Set 
\beq \label{f}
f_0(x)=J(x) \nabla H_0(x).
\eeq
\begin{lemma}
Vector field $f_0(x)$ satisfies the following identity:
\beq \label{Af'}
 A^{\rm T}f_0'(x)=f_0'(x)A^{\rm T}.
\eeq
\end{lemma}
\begin{proof}

We prove that both terms on the right-hand side of \eqref{f'} satisfy \eqref{Af'} separately. For the first term we use \eqref{d2H}:
$$
A^{\rm T}J(x) \nabla^2 H_0(x)=J(x)A\nabla^2 H_0(x)=J(x)\nabla^2 H_0(x)A^{\rm T}.
$$
For the second term, we prove a more general statement, namely
$$
A^{\rm T}(J'y)=(J'y)A^{\rm T}
$$
for an arbitrary vector 
$$
y=\begin{pmatrix} \xi \\ \eta \end{pmatrix}\in\bbR^{2n}, \quad \xi,\eta\in \bbR^n.
$$
Here, according to \eqref{P short}, $J'y$ is the $2n\times 2n$ matrix whose first $n$ columns are given by 
$J_ky$, $k=1,\ldots,n$, while the last $n$ columns vanish, so that
\beq \label{J'y}
J'y = \begin{pmatrix} J_1y & \cdots & J_n y & 0 & \cdots & 0 \end{pmatrix}  
         = \begin{pmatrix} Q_1 \eta & \cdots & Q_n \eta & 0 & \cdots & 0 \\
                                    -Q_1 \xi & \cdots & -Q_n \xi & 0 & \cdots & 0 \end{pmatrix}.
\eeq
Thus, according to \eqref{A final}, relation $A^{\rm T}(J'y)=(J'y)A^{\rm T}$ is equivalent to
\beq \label{main lemma aux1}
\cP^{\rm T}\begin{pmatrix}  Q_1 \eta & \cdots & Q_n \eta \end{pmatrix}=\begin{pmatrix}  Q_1 \eta & \cdots & Q_n \eta \end{pmatrix}\cP^{\rm T},
\eeq
\beq \label{main lemma aux2}
\cP^{\rm T}\begin{pmatrix}  Q_1 \xi & \cdots & Q_n \xi \end{pmatrix}=
                         \begin{pmatrix} Q_1 \xi & \cdots & Q_n \xi  \end{pmatrix}  \cP^{\rm T}  .
\eeq
One easily sees that
\beq \label{J'y block}
\begin{pmatrix}  Q_1 \xi & \cdots & Q_n \xi \end{pmatrix}=\xi_1I+\xi_2\cP^{-1}+\xi_3\cP^{-2}+\ldots+\xi_n\cP^{-n+1}
\eeq
is a circulant matrix, therefore it commutes with $\cP^{\rm T}=\cP^{-1}$. This proves \eqref{main lemma aux1}, \eqref{main lemma aux2} .
\end{proof}

Now let $H_1(x)$ be  a function satisfying $\nabla H_1(x)=A\nabla H_0(x)$, and set 
\beq \label{g}
f_1(x)=J(x)\nabla H_1(x)=J(x)A\nabla H_0(x).
\eeq
Due to \eqref{cond A}, we have:
\beq \label{f vs g}
f_1(x)= A^{\rm T} f_0(x).
\eeq
By differentiation of  \eqref{f vs g} we have:
\beq \label{f' vs g'}
f_1'(x)=A^{\rm T} f'_0(x).
\eeq
\begin{cor}
The following identities hold true:
\beq \label{vf comm}
 f_0'(x)f_1(x)=f_1'(x)f_0(x),
\eeq
\beq \label{f'g'=g'f'}
 f_0'(x)f_1'(x)=f_1'(x)f_0'(x).
\eeq
\end{cor}
\begin{proof}
We compute with the help of \eqref{f vs g}, \eqref{Af'}:
\[
f_0'(x)f_1(x)=f_0'(x)A^{\rm T}f_0(x)=A^{\rm T}f_0'(x)f_0(x)=f_1'(x)f_0(x),
\]
and similarly, with the help of \eqref{f' vs g'}, \eqref{Af'}:
\[
f_0'(x)f_1'(x)=f_0'(x)A^{\rm T}f_0'(x)=A^{\rm T}f_0'(x)f_0'(x)=f_1'(x)f_0'(x).
\]
Note that \eqref{vf comm} expresses the commutativity of the vector fields $f_0(x)$ and $f_1(x)$. 
\end{proof}

\begin{lemma} \label{lemma matrix polynomial}
We have:
\beq \label{Jd2H as polynomial}
\big(J(x)\nabla^2 H\big)^2=q_0(x) I+q_1(x)A+\ldots +q_{n-1}(x)A^{n-1},
\eeq
where
\begin{eqnarray*}
q_0(x) & = & \frac{1}{2n} {\rm tr} \Big(\big(J(x)\nabla^2 H\big)^2\Big),   \label{q_0}\\
q_\ell(x) & = & \frac{1}{2n} {\rm tr} \Big(\big(A^{\rm T}\big)^{\ell}\big(J(x)\nabla^2 H\big)^2\Big)\\
               &= &\frac{1}{2n}{\rm tr}\big(J(x)\nabla^2 H_\ell\cdot J(x)\nabla^2 H\big),   \quad \ell=1,\ldots, n-1. \label{q_i}
\end{eqnarray*}
\end{lemma}
\begin{proof}
One easily computes from \eqref{P full}, \eqref{cH}:
$$
J(x)\nabla^2 H=\begin{pmatrix} X\cH_2 & X\cH_3 \\ -X\cH_1 & -X\cH_2\end{pmatrix},
$$
where for the sake of brevity we write $X$ for $X(u)$, and further:
$$
\big(J(x)\nabla^2 H\big)^2=\begin{pmatrix} X\cH_2X\cH_2-X\cH_3X\cH_1 & X\cH_2X\cH_3-X\cH_3X\cH_2 \\
                                                                    -X\cH_1X\cH_2+X\cH_2X\cH_1 & -X\cH_1X\cH_3+X\cH_2X\cH_2\end{pmatrix}.
$$
By Lemma \ref{lemma ABC}, we have: $\cH_iX\cH_j=\cH_jX\cH_i$, therefore the off-diagonal blocks here vanish, while the diagonal blocks are equal. Moreover, the diagonal blocks are circulant matrices and therefore can be represented as
\beq \label{Jd2H block}
X\cH_2X\cH_2-X\cH_3X\cH_1=q_0(x) I+q_1(x)\cP+\ldots q_{n-1}(x)\cP^{n-1},
\eeq
which finishes the proof.
\end{proof}

\begin{lemma} \label{lemma grads pq}
Functions $q_\ell(x)$ from Lemma \ref{lemma matrix polynomial} satisfy
$$
\nabla q_\ell(x)=A^\ell\ \nabla q_0(x)\quad \Leftrightarrow\quad \nabla q_\ell=A\nabla q_{\ell-1}, \quad \ell=1,\ldots,n-1.
$$
\end{lemma}
\begin{proof}
According to equation \eqref{Jd2H block}, we can write:
\begin{eqnarray*}
q_0(x) & = & \frac{1}{n}\ {\rm tr}\big(X(u){\mathcal H}_2X(u){\mathcal H}_2-X(u){\mathcal H}_3X(u){\mathcal H}_1\big), \\
q_\ell(x) & = & \frac{1}{n}\ {\rm tr}\big(\cP^{-\ell}X(u){\mathcal H}_2X(u){\mathcal H}_2-\cP^{-\ell}X(u){\mathcal H}_3X(u){\mathcal H}_1\big).
\end{eqnarray*}
We denote by $u,a,b,c\in \mathbb R^n$ the first columns of the matrices $X(u)$, ${\mathcal H}_1$, ${\mathcal H}_2$, ${\mathcal H}_3$, respectively. Then the $i$-th column of the matrix $X(u)$ is $\cP^i u$, the $i$-th column of the matrix $\cP^{-\ell}X(u)$ is $\cP^{i-\ell} u$, and the entry $(i,j)$ of the circulant matrix $X(u){\mathcal H}_1$, say, is equal to $\langle u,\cP^{j-i}a\rangle$.  Therefore, 
\begin{eqnarray*}
q_0(x) & = &\frac{1}{n}\  \sum_{i=0}^{n-1}\Big(\langle u,\cP^{i}b\rangle\langle u,\cP^{-i}b\rangle -\langle u,\cP^{i}c\rangle\langle u,\cP^{-i}a\rangle\Big), \\
q_\ell(x) & = & \frac{1}{n}\ \sum_{i=0}^{n-1}\Big(\langle \cP^{-\ell}u,\cP^{i}b\rangle\langle u,\cP^{-i}b\rangle -\langle \cP^{-\ell}u,\cP^{i}c\rangle\langle u,\cP^{-i}a\rangle\Big).
\end{eqnarray*}
Now we directly compute (we write $\nabla_{u}q_\ell(x)$ for the first $n$ components of the gradients, and remember that the last $n$ components vanish: $\nabla_{v}q_\ell(x)=0$):
\begin{eqnarray*}
\nabla_{u} q_0(x) & = & \frac{1}{n}\ \sum_{i=0}^{n-1}\Big(\langle u,\cP^{-i}b\rangle \cP^{i}b+\langle u,\cP^{i}b\rangle \cP^{-i}b  -\langle u,\cP^{-i}a\rangle \cP^{i}c-\langle u,\cP^{i}c\rangle \cP^{-i}a\Big), \\
\nabla_{u} q_\ell(x) & = & \frac{1}{n}\ \sum_{i=0}^{n-1}\Big(\langle u,\cP^{-i}b\rangle \cP^{i+\ell}b +\langle u,\cP^{i+\ell}b\rangle \cP^{-i}b  -\langle u,\cP^{-i}a\rangle \cP^{i+\ell}c-\langle u,\cP^{i+\ell}c\rangle \cP^{-i}a\Big).\qquad
\end{eqnarray*}
Shifting in the second and the fourth sums on the right-hand side of the last equation index $i$ by $-\ell$, we find:
$$
\nabla_{u} q_\ell(x)= \sum_{i=0}^{n-1}\Big(\langle u,\cP^{-i}b\rangle \cP^{i+\ell}b +\langle u,\cP^{i}b\rangle \cP^{-i+\ell}b  -\langle u,\cP^{-i}a\rangle \cP^{i+\ell}c-\langle u,\cP^{i}c\rangle \cP^{-i+\ell}a\Big).
$$
Thus, we obtain $\nabla_{u} q_\ell(x)=\cP^{\ell}\nabla_{u} q_0(x)$, which, of course, yields  the required relation $\nabla q_\ell(x)=A^{\ell}\nabla q_0(x)$.
\end{proof}

\section{Associated vector fields}
\label{sect assoc}

\begin{definition} \label{def assoc}
Let the matrix 
$$
B=\alpha_0 I+\alpha_1A+\ldots +\alpha_{n-1} A^{n-1}
$$
satisfy 
$$
B^2=I.
$$
Then the vector field
$$
g(x) =J(x)B \nabla H_0(x)=B^{\rm T}J(x)\nabla H_0(x)=B^{\rm T} f_0(x)
$$
is called {\em associated} to the vector field $f_0(x)$. The vector field $g(x)$ is Hamiltonian,
$$
g(x) = J(x)\nabla K(x),
$$
with the Hamilton function
$$
K(x) =  \alpha_0 H_0(x)+\alpha_1 H_1(x)+\ldots+\alpha_{n-1} H_{n-1}(x).
$$
\end{definition}
This defines an equivalence relation on the set of vector fields $J(x)\nabla H(x)$ with the Hamilton functions $H(x)$ satisfying \eqref{d2H}. 

\begin{lemma} \label{lemma g'g'}
If vector field $g(x)$ is associated to $f_0(x)$ via the matrix $B$, then the following identities hold true:
\beq \label{f'f=g'g}
 g'(x)g(x)=f_0'(x)f_0(x),
\eeq
\beq \label{f'f'=g'g'}
(g'(x))^2=(f_0'(x))^2,
\eeq
and
\beq \label{Jd2HJd2K}
J(x)\nabla^2 H\cdot J(x)\nabla^2 K=J(x)\nabla^2 K\cdot J(x)\nabla^2 H,
\eeq
\beq \label{Jd2H2}
\big(J(x)\nabla^2 H\big)^2=\big(J(x)\nabla^2 K\big)^2,
\eeq
\beq \label{d2HJ2}
\big(\nabla^2 HJ(x)\big)^2=\big(\nabla^2 KJ(x)\big)^2.
\eeq
\end{lemma}

\begin{proof}
We first check \eqref{f'f=g'g}:
\[
g'(x)g(x)=g'(x)B^{\rm T}f_0(x)=B^{\rm T}g'(x)f_0(x)=(B^{\rm T})^2 f_0'(x)f_0(x)= f_0'(x)f_0(x).
\]
For \eqref{f'f'=g'g'} everything is similar:
\[
g'(x)g'(x)=g'(x)B^{\rm T}f_0'(x)=B^{\rm T}g'(x)f_0'(x)=(B^{\rm T})^2 (f_0'(x))^2= (f_0'(x))^2.
\]
Further, we have, according to \eqref{cond A} and to \eqref{d2H}:
$$
J(x)\nabla^2 K=J(x)B\ \nabla^2 H = B^{\rm T}J(x)\nabla^2 H = J(x)\nabla^2 HB^{\rm T}.
$$
These two formulas for $J(x)\nabla^2 K$, together with $(B^{\rm T})^2=I$, yield \eqref{Jd2HJd2K} and \eqref{Jd2H2}.
\end{proof}

%


\begin{lemma}
In dimension $2n>4$, there exist $n$ linearly independent matrices
\begin{equation*}
B_j=\alpha_0^{(j)}I+\alpha_1^{(j)} A+ \ldots + \alpha_{n-1}^{(j)} A^{n-1}\quad (j=0,\ldots,n-1),
\end{equation*}
satisfying
\begin{equation}\label{B^2}
B_j^2=I,
\end{equation}
and such that 
\begin{equation}\label{B char pol}
\det(B_j-\lambda I)=(\lambda-1)^{2n-2}(\lambda+1)^2.
\end{equation}
They are related by
\begin{equation}\label{B sum}
\sum_{j=0}^{n-1} B_j=(n-2)I.
\end{equation}
\end{lemma}
\begin{proof} The matrix $A$ from \eqref{A final} has $n$ eigenvalues $\lambda_0=1$, $\lambda_1=\omega$, $\ldots,$ $\lambda_{n-1}=\omega^{n-1}$ with $\omega=\exp(2\pi i/n)$, all of the multiplicity 2, with two linearly independent eigenvectors. So, $A$ is diagonalizable and its characteristic polynomial is of the form 
$$
\det(\lambda I-A)=(\lambda-1)^2(\lambda-\omega)^2\cdots(\lambda-\omega^{n-1})^2=(\lambda^n-1)^2.
$$
We construct (according to the Lagrange interpolating formula) $n$ polynomials of degree $n-1$,
$$
B_j(\lambda)=\alpha_0^{(j)}+\alpha_1^{(j)}\lambda+\ldots+\alpha_{n-1}^{(j)}\lambda^{n-1} \quad (j=0,\ldots,n-1)
$$
such that
\begin{equation}\label{B constr}
B_j(\lambda_j)=-1 \quad {\rm and} \quad B_j(\lambda_k)=1 \quad {\rm for} \quad k\neq j.
\end{equation}
Thus, each matrix $B_j=B_j(A)$ has $n-1$ double eigenvalues equal to 1 and one double eigenvalue equal to $-1$, so that \eqref{B char pol} is satisfied. As a corollary, each matrix $B_j^2$ has all $2n$ eigenvalues equal to 1, which, together with diagonalizabilty, yields \eqref{B^2}. Finally, from \eqref{B constr} there follows that the degree $n-1$ polynomial $B_0(\lambda)+\ldots+B_{n-1}(\lambda)$ takes the value $n-2$ at the $n$ points $\lambda_0$, $\ldots,$ $\lambda_{n-1}$, therefore it is identically equal to $n-2$. This yields \eqref{B sum}. Actually, the matrices $B_j$ can be easily computed explicitly. The $n\times n$ diagonal blocks of $B_j$ are given by
$$
I-\frac{2}{n}\begin{pmatrix}   1 & \omega^{j} &  \omega^{2j}  & \ldots & \omega^{(n-1)j}  \\
                                             \omega^{(n-1)j} & 1 & \omega^{j} & \ldots & \omega^{(n-2)j} \\
                                             \omega^{(n-2)j} & \omega^{(n-1)j} & 1 & \ldots & \omega^{(n-3)j} \\
                                             \ldots & \ldots & \ldots & \ldots & \ldots  \\
                                             \omega^{j} & \omega^{2j} & \omega^{3j} & \ldots &  1 
                  \end{pmatrix}, \quad \omega=\exp(2\pi i/n),
$$
so that $\alpha_0^{(j)}=1-\frac{2}{n}$ and $\alpha_k^{(j)}=-\frac{2}{n}\omega^{kj}$.
\end{proof}

Our main results are the following: for two associated vector fields $f$ and $g$ the Kahan maps $\Phi_f$ and $\Phi_g$ commute (Theorem \ref{th commute}), share $n$ independent integrals of motion (Theorem \ref{th integrals}), and share an invariant symplectic structure (Theorem \ref{th Poisson}).

\section{Commutativity of maps}
\label{sect commut}

\begin{theorem}\label{th commute}
Let $f(x)=J(x)\nabla H(x)$ and $g(x)=J(x)\nabla K(x)$ be two associated vector fields, via the matrix $B$. Then the maps
\bea
&& \Phi_{f}: x \mapsto \t x = 
\left( I - \ep f'(x) \right)^{-1} x =
\left( I + \ep f'(\t x) \right) x, \label{eq: Phi1} \\
&& \Phi_{g}: x \mapsto \widehat x = 
\left( I - \ep g'(x) \right)^{-1} x =
\left( I + \ep g'(\widehat x) \right) x, \label{eq: Phi2} 
\eea
commute:  $\Phi_{f} \circ \Phi_{g}=\Phi_{g} \circ \Phi_{f}$.
\end{theorem}

\begin{proof}
We have:
\begin{equation}\label{eq: hat tilde}
\left(\Phi_{g} \circ \Phi_{f}\right)(x)=\left( I - \ep g'(\t x) \right)^{-1} \left( I + \ep f'(\t x) \right) x,
\end{equation}
and
\begin{equation}\label{eq: tilde hat}
\left(\Phi_{f} \circ \Phi_{g}\right)(x)=\left( I - \ep f'(\widehat x) \right)^{-1} \left( I + \ep g'(\widehat x) \right) x.
\end{equation}
We prove the following matrix equation:
\beq \label{det}
\left( I - \ep g'(\t x) \right)^{-1} \left( I + \ep f'(\t x) \right) =\left( I - \ep f'(\widehat x) \right)^{-1} \left( I + \ep g'(\widehat x) \right),
\eeq
which is stronger than the vector equation $\left(\Phi_{f}\circ\Phi_{g}\right)(x)=\left(\Phi_{g}\circ\Phi_{f}\right)(x)$ expressing commutativity. Equation \eqref{det} is equivalent to
\beq \label{det 1}
\left( I - \ep f'(\widehat x) \right) \left( I - \ep g'(\t x) \right)^{-1} =
  \left( I + \ep g'(\widehat x) \right)\left( I + \ep f'(\t x) \right)^{-1}.
\eeq
From \eqref{f'f'=g'g'} we find:
\[
\left( I - \ep g'(\t x) \right)^{-1}=\left( I + \ep g'(\t x) \right)\left( I - \ep^2 (f'(\t x))^2 \right)^{-1},
\]
\[
\left( I + \ep f'(\t x) \right)^{-1}=\left( I - \ep f'(\t x) \right)\left( I - \ep^2(f'(\t x))^2 \right)^{-1}.
\]
With this at hand, equation \eqref{det 1} is equivalent to
\[
\left( I - \ep f'(\widehat x) \right)\left( I + \ep g'(\t x) \right)=
\left( I + \ep g'(\widehat x) \right)\left( I - \ep f'(\t x) \right).
\]
Here the quadratic in $\epsilon$ terms cancel by virtue of \eqref{f' vs g'} and \eqref{Af'}:
\[
f'(\widehat x)g'(\t x) = f'(\widehat x)B^{\rm T}f'(\t x) =  B^{\rm T}f'(\widehat x)f'(\t x)  =  g'(\widehat x)f'(\t x),
\]
so that we are left with the terms linear in $\epsilon$:
\beq \label{aux}
-  f'(\widehat x)+ g'(\t x)= g'(\widehat x)- f'(\t x).
\eeq
Since the tensors $f''$, $g''$ are constant, we have:
\[
f'(\widehat x)=f'(x)+f''(\widehat x-x)=f'(x)+2\ep f''\left(I-\ep g'(x)\right)^{-1}g(x),
\]
\[
g'(\widehat x)=g'(x)+g''(\widehat x-x)=g'(x)+2\ep g''\left(I-\ep g'(x)\right)^{-1}g(x),
\]
\[
f'(\t x)=f'(x)+f''(\t x-x)=f'(x)+2\ep f''\left(I-\ep f'(x)\right)^{-1}f(x),
\]
\[
g'(\t x)=g'(x)+g''(\t x-x)=g'(x)+2\ep g''\left(I-\ep f'(x)\right)^{-1}f(x).
\]
Thus, equation \eqref{aux} is equivalent to 
\begin{align} \label{aux2}
& f''\left(I-\ep g'(x)\right)^{-1}g(x)+g''\left(I-\ep g'(x)\right)^{-1}g(x) = \nonumber\\
& \qquad f''\left(I-\ep f'(x)\right)^{-1}f(x)+g''\left(I-\ep f'(x)\right)^{-1}f(x).
\end{align}
At this point, we use the following statement.

\begin{lemma} \label{lemma g''}
For any vector $v\in\mathbb C^{2n}$ we have:
\beq \label{g''}
 g''(x)v=f''(x)(B^{\rm T}v), \quad  f''(x)v=g''(x)(B^{\rm T}v).
\eeq
\end{lemma}

\noindent
We compute the matrices on the left-hand side of \eqref{aux2} with the help of \eqref{g''}, \eqref{f vs g}, \eqref{f' vs g'}:
\begin{eqnarray*}
f''\left(I-\ep g'(x)\right)^{-1}g(x) & = & f''\left(I-\ep^2(f'(x))^2\right)^{-1}\left(g(x)+\ep g'(x)g(x)\right),\\
g''\left(I-\ep g'(x)\right)^{-1}g(x) & = & f''\left(I-\ep^2(f'(x))^2\right)^{-1}B^{\rm T}\left(g(x)+\ep g'(x)g(x)\right)\\
                                                              & = & f''\left(I-\ep^2 (f'(x))^2\right)^{-1}\left(f(x)+\ep f'(x)g(x)\right),
\end{eqnarray*}
and similarly    
\begin{eqnarray*}
f''\left(I-\ep f'(x)\right)^{-1}f(x) & = & f''\left(I-\ep^2 (f'(x))^2\right)^{-1}\left(f(x)+\ep f'(x)f(x)\right)\\
g''\left(I-\ep f'(x)\right)^{-1}f(x) & = &f''\left(I-\ep^2 (f'(x))^2\right)^{-1}B^{\rm T}\left(f(x)+\ep f'(x)f(x)\right)\\
                                                              & = & f''\left(I-\ep^2 (f'(x))^2\right)^{-1}\left(g(x)+\ep g'(x)f(x)\right).
\end{eqnarray*}                                                          
Collecting all the results and using \eqref{vf comm} and \eqref{f'f=g'g}, we see that the proof is complete. 
\end{proof}

{\em Proof of Lemma \ref{lemma g''}.} The identities in question are equivalent to
\beq \label{Af''}
 B^{\rm T}(f''(x)v)=f''(x)(B^{\rm T}v), \quad  B^{\rm T}(g''(x)v)=g''(x)(B^{\rm T}v).
\eeq
(Actually, both tensors $f''$ and $g''$ are constant, i.e., do not depend on $x$.) To prove the latter identities, we start with equation \eqref{Af'} written in components:
\[
\sum_k (B^{\rm T})_{ik}\frac{\partial f_k}{\partial x_\ell}=\sum_k \frac{\partial f_i}{\partial x_k}(B^{\rm T})_{k\ell}.
\]
Differentiating with respect to $x_j$, we get:
\[
\sum_k (B^{\rm T})_{ik}\frac{\partial^2 f_k}{\partial x_j\partial x_\ell}=\sum_k \frac{\partial f_i}{\partial x_j\partial x_k}(B^{\rm T})_{k\ell}.
\]
Hence,
\[
\sum_{k,\ell} (B^{\rm T})_{ik}\frac{\partial^2 f_k}{\partial x_j\partial x_\ell}v_\ell=\sum_{k,\ell} \frac{\partial f_i}{\partial x_j\partial x_k}(B^{\rm T})_{k\ell}v_\ell,
\]
which is nothing but the $(i,j)$ entry of the matrix identity \eqref{Af''}. \qed

\section{Integrals of motion}
\label{sect integrals}

\begin{theorem} \label{th integrals}
Let $f(x)=J(x)\nabla H(x)$ and $g(x)=J(x)\nabla K(x)$ be two associated vector fields, via the matrix $B$. 
Then the maps $\Phi_{f}$ and $\Phi_{g}$ share two functionally independent conserved quantities 
\begin{equation}\label{tH}
\t H(x,\epsilon)=h(x,\Phi_f(x,\epsilon), \epsilon),
\end{equation}
and 
\begin{equation}\label{tK}
\t K(x,\epsilon)=k(x, \Phi_g(x,\epsilon),\epsilon),
\end{equation}
where
\begin{equation}\label{h}
h(x,\t x, \epsilon)=(2\epsilon)^{-1}x^{\rm T}J^{-1}\Big( \frac{x+\t x}{2}\Big) \t x,
\end{equation}
and 
\begin{equation}\label{k}
k(x, \widehat x,\epsilon)=(2\epsilon)^{-1} x^{\rm T}
J^{-1} \Big(\frac{x+\widehat x}{2}\Big)\widehat x.
\end{equation}
\end{theorem}
\begin{proof}
First, we show that $\t H(x,\epsilon)$ is an integral of motion for the map $\Phi_f$. We start with giving several equivalent formulas for $\t H(x,\epsilon)$. Upon using the skew-symmetry of $J(x)$ and formula \eqref{Phi x-x}, we can rewrite \eqref{tH} as
\begin{eqnarray*}
\t H(x,\epsilon) & = & (2\epsilon)^{-1}x^{\rm T}J^{-1}\Big( \frac{x+\t x}{2}\Big) (\t x-x) \nonumber \\
 & = &  x^{\rm T}J^{-1}\Big( \frac{x+\t x}{2}\Big)\Big(I-\epsilon J(x)\nabla^2 H \Big)^{-1} J\Big( \frac{x+\t x}{2}\Big) \nabla H(x).
\end{eqnarray*}
\begin{lemma} \label{lemma comm}
For any $x,y\in\bbR^n$, we have:
\beq
J(x)\nabla^2 H \ J(y)=J(y)\nabla^2 H \ J(x).
\eeq
\end{lemma}
\begin{proof}  With notation \eqref{P full}, \eqref{cH}, we have to prove:
$$
X{\mathcal H}_kY=Y{\mathcal H}_k X, \quad k=1,2,3.
$$
But this follows directly from Lemma \ref{lemma ABC}.
\end{proof}

From this lemma, there follows:
$$
\Big(I-\epsilon J(x)\nabla^2 H \Big)^{-1} J\Big( \frac{x+\t x}{2}\Big)=J\Big( \frac{x+\t x}{2}\Big)\Big(I-\epsilon \nabla^2 H\ J(x) \Big)^{-1} ,
$$
and 
\begin{equation*}
\t H(x,\epsilon)=x^{\rm T}\Big(I-\epsilon \nabla^2 H\ J(x) \Big)^{-1}\nabla H(x).
\end{equation*}
Expanding into a power series in $\epsilon$, we find:
\begin{eqnarray*}
\t H(x,\epsilon) & = &  x^{\rm T}\sum_{k=0}^\infty \epsilon^{k}\Big(\nabla^2 H \ J(x)\cdots \nabla^2 H\ J(x)\Big) \nabla H(x) \\
  & = & x^{\rm T}\sum_{k=0}^\infty \epsilon^{k}\Big(\nabla^2 H \ J(x)\cdots \nabla^2 H\ J(x)\nabla^2 H\Big) x.
\end{eqnarray*}
The matrix in the parentheses involves $k+1$ times $\nabla^2 H$ and $k$ times $J(x)$, therefore it is symmetric if $k$ is even, and skew-symmetric if $k$ is odd. Therefore, all terms with odd $k$ vanish. We have the following equivalent expressions:
\begin{eqnarray}
\t H(x,\epsilon) & = & x^{\rm T} (\nabla^2 H)\Big(I-\epsilon J(x) \nabla^2 H\Big)^{-1} x     \label{tH 1}\\
& = & x^{\rm T}(\nabla^2) H \Big(I-\epsilon^2 \big(J(x) \nabla^2 H\big)^2\Big)^{-1} x        \label{tH 2}\\
& = & x^{\rm T} \Big(I-\epsilon \nabla^2 H\ J(x) \Big)^{-1} (\nabla^2 H) x                           \label{tH 3}\\
& = & x^{\rm T} \Big(I-\epsilon^2 \big(\nabla^2 H\ J(x)\big)^2\Big)^{-1} (\nabla^2 H) x.      \label{tH 4}
\end{eqnarray}
Moreover, in \eqref{tH 1} and \eqref{tH 3} one can replace $\epsilon$ by $-\epsilon$.

The fact that $\t H(x,\epsilon)$ is an even function of $\epsilon$ ensures that it is an integral of $\Phi_f$. Indeed, by virtue of \eqref{eq: reversible} 
we have: 
$$
\t H(x,-\epsilon)=h(x,\Phi_f(x,-\epsilon), -\epsilon)=h(x,\Phi_f^{-1}(x,\epsilon),-\epsilon)=h(\Phi_f^{-1}(x,\epsilon),x,\epsilon). 
$$
The last equality follows from the property of the function $h$,
$$
h(x,y,\epsilon)=h(y,x,-\epsilon),
$$
which follows from the definition \eqref{h} by the skew-symmetry of the matrix $J$. Thus, if $\t H(x,\epsilon)=\t H(x,-\epsilon)$, then
$$
\t H(x,\epsilon)=\t H(\Phi_f^{-1}(x,\epsilon),\epsilon),
$$
which proves the claim. 

Next, we show that $\t K(x,\epsilon)$ also is an integral of motion for the map $\Phi_f$. For this goal, we first compute, based on \eqref{tH 4}:
$$
\t K(\t x,\epsilon) =  \t x^{\rm T}\Big(I-\ep^2 \big(\nabla^2 K\ J(\t x)\big)^2\Big)^{-1}(\nabla^2 K)\t x .
$$
By virtue of \eqref{Phi x} we have:
\begin{eqnarray*}
\t K(\t x,\epsilon) & = &  x^{\rm T}\Big(I+\epsilon J(\t x)\nabla^2 H\Big)^{\rm T}\Big(I-\epsilon J(x)\nabla^2 H\Big)^{\rm -T}
                                      \Big(I-\ep^2 \big(\nabla^2 K\ J(\t x)\big)^2\Big)^{-1}\times \\
                           &    & \times (\nabla^2 K)\Big(I-\epsilon J(x)\nabla^2 H\Big)^{-1} \Big(I+\epsilon J(\t x)\nabla^2 H\Big)x .
\end{eqnarray*}
By Lemma \ref{lemma comm} we have:
\begin{eqnarray*}
\t K(\t x,\epsilon) & = &  x^{\rm T}\Big(I+\epsilon\nabla^2 H\  J(x)\Big)^{-1}\Big(I-\epsilon \nabla^2 H\ J(\t x)\Big)
                                      \Big(I-\ep^2 \big(\nabla^2 K\ J(\t x)\big)^2\Big)^{-1}\times \\
                           &    & \times (\nabla^2 K) \Big(I+\epsilon J(\t x)\nabla^2 H\Big)\Big(I-\epsilon J(x)\nabla^2 H\Big)^{-1}x .
\end{eqnarray*}
Next, we find:
\begin{eqnarray*}
(\nabla^2 K) \Big(I+\epsilon J(\t x)\nabla^2 H\Big) & = & B(\nabla^2 H) \Big(I+\epsilon J(\t x)\nabla^2 H\Big)\\
          & = & B\Big(I+\epsilon \nabla^2 H\ J(\t x)\Big)(\nabla^2 H)\\
          & = & \Big(I+\epsilon \nabla^2 H\ J(\t x)\Big)B(\nabla^2 H)\\
          & = & \Big(I+\epsilon \nabla^2 H\ J(\t x)\Big)(\nabla^2 K) .
\end{eqnarray*}
Here, the last but one equality is justified as follows: 
$$
B\nabla^2 H\ J(\t x)=\nabla^2 H\ B^{\rm T}J(\t x)=\nabla^2 H\ J(\t x)B.
$$
Similarly, we find:
$$
(\nabla^2 K)\Big(I-\epsilon J(x)\nabla^2 H\Big)^{-1}=\Big(I-\epsilon \nabla^2 H\ J(x)\Big)^{-1}(\nabla^2 K).
$$
Collecting all the results, we have:
\begin{eqnarray*}
\t K(\t x,\epsilon) & = &  x^{\rm T}\Big(I+\epsilon\nabla^2 H\  J(x)\Big)^{-1}\Big(I-\epsilon \nabla^2 H\ J(\t x)\Big)
                                      \Big(I-\ep^2 \big(\nabla^2 K\ J(\t x)\big)^2\Big)^{-1}\times \\
                           &    & \times  \Big(I+\epsilon \nabla^2 H\ J(\t x)\Big)\Big(I-\epsilon \nabla^2 H\ J(x)\Big)^{-1}(\nabla^2 K)x.
\end{eqnarray*}                           
Applying equation \eqref{d2HJ2}  twice, we find:                    
\begin{eqnarray*}
\t K(\t x,\epsilon) & = &  x^{\rm T}\Big(I+\epsilon\nabla^2 H\  J(x)\Big)^{-1}\Big(I-\epsilon \nabla^2 H\ J(\t x)\Big)
                                      \Big(I-\ep^2 \big(\nabla^2 H\ J(\t x)\big)^2\Big)^{-1}\times \\
                           &    & \times  \Big(I+\epsilon \nabla^2 H\ J(\t x)\Big)\Big(I-\epsilon \nabla^2 H\ J(x)\Big)^{-1}(\nabla^2 K)x\\
                           & = & x^{\rm T}\Big(I-\epsilon^2\big(\nabla^2 H\  J(x)\big)^2\Big)^{-1}(\nabla^2 K)x\\
                           & = & x^{\rm T}\Big(I-\epsilon^2\big(\nabla^2 K\  J(x)\big)^2\Big)^{-1}(\nabla^2 K)x\\
                           & = & \t K(x,\epsilon),
\end{eqnarray*}
which finishes the proof.
\end{proof}

\section{Invariant Poisson structure}
\label{sect symplectic}

\begin{theorem} \label{th Poisson}
Let $f(x)=J(x)\nabla H(x)$ and $g(x)=J(x)\nabla K(x)$ be two associated vector fields, via the matrix $B$.  Then
both maps $\Phi_f$ and $\Phi_g$ are Poisson with respect to the brackets with the Poisson tensor $\Pi(x)$ given by
\begin{eqnarray} \label{Pi short}
\Pi(x)  & = & J(x)- \epsilon^2 J(x) \cdot \nabla^2 H\cdot J(x) \cdot \nabla^2 H\cdot J(x).
\end{eqnarray}
\end{theorem}

This theorem is a direct consequence of the following two statements combined with Lemma \ref{lemma grads pq}.

\begin{proposition} \label{prop Poisson formula}
For the matrix $\Pi(x)$ from \eqref{Pi short}, we have:
\beq \label{Pois prop}
d\Phi_f(x)\Pi(x) (d\Phi_f(x))^{\rm T}=\Pi(\t x) .
\eeq
\end{proposition}
\begin{proposition} \label{prop Poisson property}
A matrix 
\beq
\Pi(x)=(1-\epsilon^2 q_0(x))J(x)-\epsilon^2\sum_{\ell=1}^{n-1} q_\ell(x) A^\ell J(x)
\eeq
is a Poisson tensor if and only if the functions $q_\ell(x)$ satisfy
\beq
\nabla q_\ell(x) =A\nabla q_{\ell-1}(x), \quad \ell\in \mathbb Z/(n\mathbb Z),
\eeq
or, equivalently,
\beq \label{dqdx}
\frac{\partial q_\ell}{\partial x_i}=\frac{\partial q_{\ell-1}}{\partial x_{i+1}}, \quad \ell\in\mathbb Z/(n\mathbb Z), \quad i\in [1,2n],
\eeq
where the latter equation for $i=n$ and for $i=2n$ should be read as 
$$
\frac{\partial q_\ell}{\partial x_n}=\frac{\partial q_{\ell-1}}{\partial x_1}, \quad {resp.} \quad 
\frac{\partial q_\ell}{\partial x_{2n}}=\frac{\partial q_{\ell-1}}{\partial x_{n+1}}.
$$
\end{proposition}

\noindent
{\em Proof of Proposition \ref{prop Poisson formula}.} With expression \eqref{Jac} for $d\Phi_f(x)$, equation \eqref{Pois prop} turns into
\beq  \label{proof Poisson 1}
\big(I+\epsilon f'(\t x)\big)\Pi(x)\big(I+\epsilon f'(\t x)\big)^{\rm T}=\big(I-\epsilon f'(x)\big)\Pi(\t x)\big(I-\epsilon f'(x)\big)^{\rm T}.
\eeq
We have: $\Pi(x)=\big(I-\epsilon^2(J(x)\nabla^2 H)^2\big)J(x)$. According to Lemma \ref{lemma matrix polynomial}, $I-\epsilon^2(J(x)\nabla^2 H)^2$ is a matrix polynomial of $A$. By virtue of \eqref{Af'}, this matrix commutes with $f'(\t x)$ (actually, with $f'$ evaluated at any point). Therefore,  equation \eqref{proof Poisson 1} is equivalent to
\begin{eqnarray} \label{proof Poisson 2}
\lefteqn{
\big(I-\epsilon^2(J(x)\nabla^2 H)^2\big)\big(I+\epsilon f'(\t x)\big)J(x)\big(I+\epsilon f'(\t x)\big)^{\rm T}} \nonumber\\
& = & 
\big(I-\epsilon^2(J(\t x)\nabla^2 H)^2\big)\big(I-\epsilon f'(x)\big)J(\t x)\big(I-\epsilon f'(x) \big)^{\rm T}.
\end{eqnarray}
\begin{lemma} \label{lemma f' vs Jd2H}
We have:
\beq \label{lemma f' vs Jd2H 1}
I-\epsilon f'(x)=\big(I-\epsilon J(x)\nabla^2 H\big)\begin{pmatrix} XX_1^{-1} & 0 \\ -X_2X_1^{-1} & I \end{pmatrix},
\eeq
and
\beq \label{lemma f' vs Jd2H 2}
I+\epsilon f'(\t x)=\big(I+\epsilon J(\t x)\nabla^2 H\big)\begin{pmatrix} \t XX_1^{-1} & 0 \\  X_2X_1^{-1} & I \end{pmatrix},
\eeq
where
\beq \label{Xs}
X=X(u), \quad \t X=X(\t u), \quad X_1=X\Big(\frac{u+\t u}{2}\Big), \quad X_2=X\Big(\frac{\t v-v}{2}\Big).
\eeq
\end{lemma}

With this lemma and Lemma \ref{lemma comm}, according to which matrices $J(x)\nabla^2 H$ and $J(\t x)\nabla^2 H$ commute, we can rewrite \eqref{proof Poisson 2} as 
\begin{eqnarray} \label{proof Poisson 3}
\lefteqn{
\big(I+\epsilon J(x)\nabla^2 H\big)\begin{pmatrix} \t XX_1^{-1} & 0 \\  X_2X_1^{-1} & I \end{pmatrix}J(x)
\begin{pmatrix} X_1^{-1}\t X & X_1^{-1}X_2 \\  0 & I \end{pmatrix}\big(I-\epsilon \nabla^2 H\ J(\t x)\big)} \nonumber\\
& = & 
\big(I-\epsilon J(\t x)\nabla^2 H\big)\begin{pmatrix} XX_1^{-1} & 0 \\ -X_2X_1^{-1} & I \end{pmatrix}J(\t x)
\begin{pmatrix} X_1^{-1}X & -X_1^{-1}X_2 \\ 0 & I \end{pmatrix}\big(I+\epsilon \nabla^2 H\ J(x)\big).\qquad
\end{eqnarray}
In the following computation, we repeatedly use the property of cyclic Hankel matrices formulated in Lemma \ref{lemma ABC}. We compute:
\begin{eqnarray*}
\lefteqn{\begin{pmatrix} \t XX_1^{-1} & 0 \\  X_2X_1^{-1} & I \end{pmatrix}J(x)
\begin{pmatrix} X_1^{-1}\t X & X_1^{-1}X_2 \\  0 & I \end{pmatrix}}\\
 & = & \begin{pmatrix} 0 & \t XX_1^{-1}X \\ -XX_1^{-1}\t X & -XX_1^{-1}X_2+X_2X_1^{-1}X\end{pmatrix} \\
 & = & \begin{pmatrix} 0 & \t XX_1^{-1}X \\ -XX_1^{-1}\t X &  0 \end{pmatrix} = \begin{pmatrix} 0 & Y \\ -Y &  0 \end{pmatrix}, 
\end{eqnarray*}
where $$Y=\t XX_1^{-1}X=XX_1^{-1}\t X$$ is a cyclic Hankel matrix (according to Lemma \ref{lemma ABC}). Similarly, 
\begin{eqnarray*}
\lefteqn{\begin{pmatrix} XX_1^{-1} & 0 \\ -X_2X_1^{-1} & I \end{pmatrix}J(\t x)
\begin{pmatrix} X_1^{-1}X & -X_1^{-1}X_2 \\ 0 & I \end{pmatrix}}\\
 & = & \begin{pmatrix} 0 & XX_1^{-1}\t X \\ -\t XX_1^{-1}X & -\t XX_1^{-1}X_2+X_2X_1^{-1}\t X\end{pmatrix} \\
 & = & \begin{pmatrix} 0 & XX_1^{-1}\t X \\ -\t XX_1^{-1}X &  0 \end{pmatrix} =  \begin{pmatrix} 0 & Y \\ -Y &  0 \end{pmatrix}.
 \end{eqnarray*}
 Observe that 
 $$
  \begin{pmatrix} 0 & Y \\ -Y &  0 \end{pmatrix}=J(x)J^{-1}\Big(\frac{x+\t x}{2}\Big)J(\t x)=J(\t x)J^{-1}\Big(\frac{x+\t x}{2}\Big)J(x).
 $$
Now, equation \eqref{proof Poisson 3} takes the form 
\begin{eqnarray} \label{proof Poisson 4}
\lefteqn{
\big(I+\epsilon J(x)\nabla^2 H\big) \begin{pmatrix} 0 & Y \\ -Y &  0 \end{pmatrix}\big(I-\epsilon \nabla^2 H\ J(\t x)\big)} \nonumber\\
& = & 
\big(I-\epsilon J(\t x)\nabla^2 H\big) \begin{pmatrix} 0 & Y \\ -Y &  0 \end{pmatrix}\big(I+\epsilon \nabla^2 H\ J(x)\big),
\end{eqnarray}
which is obviously true due to properties of cyclic Hankel matrices. 
\qed
\medskip

{\em Proof of Lemma \ref{lemma f' vs Jd2H}.} Observe that \eqref{lemma f' vs Jd2H 2} is obtained from \eqref{lemma f' vs Jd2H 1} upon replacing $\epsilon$ by $-\epsilon$ (which is equivalent to replacing $\t x$ by $\undertilde{x}$) with a subsequent shift in time. Therefore, it is sufficient to prove \eqref{lemma f' vs Jd2H 1}. The latter formula can be equivalently rewritten as
$$
\big(I-\epsilon J(x)\nabla^2 H\big)^{-1}\big(I-\epsilon f'(x)\big)=\begin{pmatrix} XX_1^{-1} & 0 \\ -X_2X_1^{-1} & I \end{pmatrix},
$$
or
$$
\big(I-\epsilon J(x)\nabla^2 H\big)^{-1}\big(I-\epsilon f'(x)\big)J\Big(\frac{x+ \t x}{2}\Big)=\begin{pmatrix} 0 & X \\  -X_1 & -X_2 \end{pmatrix},
$$
or 
$$
\bigg(\big(I-\epsilon J(x)\nabla^2 H\big)^{-1}\big(I-\epsilon f'(x)\big)-I\bigg)J\Big(\frac{x+ \t x}{2}\Big)=\begin{pmatrix} 0 & X-X_1 \\ 0 & -X_2 \end{pmatrix}.
$$
Taking into account equation \eqref{f'} and definitions \eqref{Xs}, we put the latter equation into the form
$$
\renewcommand{\arraystretch}{1.5}
\epsilon\big(I-\epsilon J(x)\nabla^2 H\big)^{-1} J'\nabla H(x)J\Big(\frac{x+ \t x}{2}\Big)=
  \begin{pmatrix} 0 & X\big(\frac{\t u -u}{2} \big) \\ 0 & X\big(\frac{\t v -v}{2} \big) \end{pmatrix}.
$$
It remains to observe that, according to \eqref{J'y}, \eqref{J'y block}, the latter equation is equivalent to (consists of $n$ cyclically shifted versions of) 
$$
\renewcommand{\arraystretch}{1.5}
\epsilon\big(I-\epsilon J(x)\nabla^2 H\big)^{-1} J\Big(\frac{x+ \t x}{2}\Big)\nabla H(x)=
  \begin{pmatrix} \frac{\t u -u}{2}  \\ \frac{\t v -v}{2} \end{pmatrix}=\frac{\t x-x}{2},
$$
which is nothing but \eqref{Phi x-x}.
\qed
\medskip


\noindent {\em Proof of Proposition \ref{prop Poisson property}.} Like in Section \ref{sect int fam}, we have to check the Jacobi identity \eqref{Jacobi id}. Again, due to the block-diagonal structure of the matrix $\Pi(x)$, one only has to check this for the cases $i,j\in\{1,\ldots,n\}$, $k\in\{n+1,\ldots,2n\}$ and $i,j\in\{n+1,\ldots,2n\}$, $k\in\{1,\ldots,n\}$, where \eqref{Jacobi id} simplifies to \eqref{Jacobi id block}. We compute:
$$
\{x_j,x_k\}=(1-\epsilon^2 q_0(x))J_{jk}(x)-\epsilon^2\sum_{\ell=1}^{n-1} q_\ell(x) (A^\ell J(x))_{jk}.
$$
Taking into account that $J(x)$ only depends on $x_m$ with $1\le m \le n$, we have:
\begin{eqnarray}
\{x_i,\{x_j,x_k\}\} & = & (1-\epsilon^2 q_0(x))\sum_{m=1}^n\frac{\partial{J_{jk}(x)}}{\partial x_m}\{x_i,x_m\}-\epsilon^2\sum_{\ell=1}^{n-1} q_\ell(x) \sum_{m=1}^n\frac{\partial (A^\ell J(x))_{jk}}{\partial x_m}\{x_i,x_m\} \nonumber\\
 &   & -\epsilon^2\sum_{\ell=0}^{n-1} \sum_{m=1}^{2n} \frac{\partial q_\ell(x)}{\partial x_m} (A^\ell J(x))_{jk}\{x_i,x_m\}. \label{proof Jacobi 1}
\end{eqnarray}
We first deal with the terms from the first line. They can be represented as
\begin{eqnarray}
 &  & (1-\epsilon^2 q_0(x))^2\sum_{m=1}^n\frac{\partial{J_{jk}(x)}}{\partial x_m}J_{im}(x) \nonumber\\
 &    & -\epsilon^2\sum_{p=1}^{n-1} (1-\epsilon^2 q_0(x))q_p(x) \sum_{m=1}^n \frac{\partial{J_{jk}(x)}}{\partial x_m}(A^p J(x))_{im} 
 \nonumber\\
 &    & -\epsilon^2\sum_{\ell=1}^{n-1} q_\ell(x)(1-\epsilon^2 q_0(x)) \sum_{m=1}^n \frac{\partial (A^\ell J(x))_{jk}}{\partial x_m}J_{im}(x) 
 \nonumber\\
 &    & +\epsilon^4 \sum_{\ell=1}^{n-1}\sum_{p=1}^{n-1} q_\ell(x) q_p(x) \sum_{m=1}^n\frac{\partial (A^\ell J(x))_{jk}}{\partial x_m} (A^p J(x))_{im}\ .
 \label{proof Jacobi 2}
 \end{eqnarray}
The contribution of these terms to \eqref{Jacobi id block} vanishes due to the following computation:
\begin{eqnarray*}
\sum_{m=1}^n\frac{\partial{(A^\ell J(x))_{jk}}}{\partial x_m}(A^p J(x))_{im} 
& = & -\sum_{m=1}^n\frac{\partial{x_{j+k+\ell-1\, ({\rm mod}\, n)}}}{\partial x_m}x_{i+m+p-1\, ({\rm mod}\, n)}\\
 & = & -\sum_{m=1}^n\delta_{m,j+k+\ell-1\, ({\rm mod}\, n)}x_{i+m+p-1\, ({\rm mod}\, n)}\\
 & = & -x_{i+j+k+\ell+p-2\, ({\rm mod}\, n)}.
\end{eqnarray*} 
This expression is symmetric with respect to $i\leftrightarrow j$, which results in a zero contribution to  \eqref{Jacobi id block}. Note this result for $\ell=p=0$ is equivalent to the Jacobi identity for the bracket $\{\cdot,\cdot\}_J$, while for $\ell=p\neq 0$ it is equivalent to the Jacobi identity for the bracket $\{\cdot,\cdot\}_{A^\ell J}$. The general result is equivalent to the compatibility of the brackets $\{\cdot,\cdot\}_{A^\ell J}$ and $\{\cdot,\cdot\}_{A^p J}$. 

The terms from the second line in \eqref{proof Jacobi 1} can be represented as 
\begin{eqnarray}
\lefteqn{-\epsilon^2\sum_{\ell=0}^{n-1} \sum_{m=1}^{2n} \frac{\partial q_\ell(x)}{\partial x_m} (A^\ell J(x))_{jk}J_{im}(x)}  \nonumber\\
 &    & +\epsilon^4 \sum_{\ell=0}^{n-1}\sum_{p=0}^{n-1} \sum_{m=1}^{2n}\frac{\partial q_\ell(x)}{\partial x_m} q_p(x) (A^\ell J(x))_{jk} (A^p J(x))_{im}\ .
 \label{proof Jacobi 3}
 \end{eqnarray}
 Due to the block structure of the matrix $J(x)$, if $i,j\in\{1,\ldots,n\}$, one can restrict the summation index $m$ here to the range $m\in\{n+1,\ldots,2n\}$, while if $i,j\in\{n+1,\ldots,2n\}$, one can restrict the summation index $m$ to  the range $m\in\{1,\ldots,n\}$. Both possibilities are considered analogously, therefore we concentrate on the second one. Thus, assume  that $i,j\in\{n+1,\ldots,2n\}$. Then the terms in \eqref{proof Jacobi 3}  quadratic in $\epsilon$ can be transformed as follows: 
\begin{eqnarray*}
\lefteqn{-\epsilon^2\sum_{\ell=0}^{n-1} \sum_{m=1}^n \frac{\partial q_\ell(x)}{\partial x_m} (A^\ell J(x))_{jk}J_{im}(x)}  \nonumber \\
 &   = & -\epsilon^2\sum_{\ell=0}^{n-1} \sum_{m=1}^n \frac{\partial q_\ell(x)}{\partial x_m} x_{j+k+\ell-1\, ({\rm mod}\, n)}x_{i+m-1\, ({\rm mod}\, n)} 
 \nonumber \\
  &   = & -\epsilon^2\sum_{a=1}^{n} \sum_{b=1}^n \frac{\partial q_{a-j-k+1\, ({\rm mod}\, n)}(x)}{\partial x_{b-i+1\, ({\rm mod}\, n)}} x_{a}x_{b}.
  \label{proof Jacobi 4}
 \end{eqnarray*}
Similarly, the terms in \eqref{proof Jacobi 3} of degree 4 in $\epsilon$ are transformed as follows:
\begin{eqnarray*}
\lefteqn{\epsilon^4 \sum_{\ell=0}^{n-1}\sum_{p=0}^{n-1} \sum_{m=1}^n\frac{\partial q_\ell(x)}{\partial x_m} q_p(x) (A^\ell J(x))_{jk} (A^p J(x))_{im}} 
\nonumber \\
 &   = & \epsilon^4\sum_{\ell=0}^{n-1}\sum_{p=0}^{n-1} \sum_{m=1}^n \frac{\partial q_\ell(x)}{\partial x_m} q_p(x)x_{j+k+\ell-1\, ({\rm mod}\, n)}x_{i+m+p-1\, ({\rm mod}\, n)} \nonumber \\
  &   = & \epsilon^4\sum_{a=1}^{n} \sum_{b=1}^n\sum_{p=0}^{n-1} \frac{\partial q_{a-j-k+1\, ({\rm mod}\, n)}(x)}{\partial x_{b-i-p+1\, ({\rm mod}\, n)}} q_p(x)x_{a}x_{b}.
 \end{eqnarray*}
Thus, we arrive at the following expressions for the quantities $\pi_{ijk}$ in \eqref{Jacobi id block} in the case $i,j\in\{n+1,\ldots,2n\}$ and $k\in\{1,\ldots,n\}$:
\begin{eqnarray}
\pi_{ijk} & = &  -\epsilon^2\sum_{a=1}^{n} \sum_{b=1}^n \left(\frac{\partial q_{a-j-k+1\, ({\rm mod}\, n)}(x)}{\partial x_{b-i+1\, ({\rm mod}\, n)}}-
 \frac{\partial q_{a-i-k+1\, ({\rm mod}\, n)}(x)}{\partial x_{b-j+1\, ({\rm mod}\, n)}}\right) x_{a}x_{b} \nonumber\\
   & &  +\epsilon^4\sum_{a=1}^{n} \sum_{b=1}^n \sum_{p=0}^{n-1}  \left(\frac{\partial q_{a-j-k+1\, ({\rm mod}\, n)}(x)}{\partial x_{b-i-p+1\, ({\rm mod}\, n)}}-
 \frac{\partial q_{a-i-k+1\, ({\rm mod}\, n)}(x)}{\partial x_{b-j-p+1\, ({\rm mod}\, n)}}\right) q_p(x)x_{a}x_{b}. \qquad
    \label{proof Jacobi 5}
\end{eqnarray}
We mention that in the case $i,j\in\{1,\ldots,n\}$ and $k\in\{n+1,\ldots,2n\}$, the expression for $\pi_{ijk}$ is almost literally the same, but with the indices $b-i+1\, ({\rm mod}\, n)$ etc. replaced by their $({\rm mod}\, n)$ representatives in the interval $[n+1,2n]$.

It remains to observe that \eqref{proof Jacobi 5} is equal to zero by virtue of \eqref{dqdx}.  \qed

\section{Differential equations for the conserved quantities of maps $\Phi_f$, $\Phi_g$}
\label{sect diff eqs}

\begin{theorem} \label{th diff eqs}
Let $f(x)=J(x)\nabla H_0(x)$ and $g(x)=J(x)\nabla K_0(x)$ be two associated vector fields, via the matrix $B$. 
Then the rational functions $\t H_0(x,\epsilon)$ and $\t K_0(x,\epsilon)$ are related by the same first order differential equation as the quadratic polynomials $H_0(x)$ and $K_0(x)$:
\begin{equation}\label{grad tK}
\nabla \t K_0(x,\epsilon) =B \nabla \t H_0(x,\epsilon).
\end{equation}
As a consequence, they satisfy the same second order differential equation \eqref{d2H} as the polynomials $H_0(x)$ and $K_0(x)$.
\end{theorem}
\begin{proof}
We start the proof with the derivation of a convenient formula for $\t H_0(x,\epsilon)$.
From \eqref{tH 4} we have:
\begin{equation} \label{tH 5}
\t H_0(x,\epsilon) = x^{\rm T} \left(I-\epsilon^2(\nabla^2 H_0\ J(x))^2\right)^{-1} \nabla H_0(x).
\end{equation}
\begin{lemma} \label{lemma inv}
We have:
\begin{equation} \label{inv}
 \left(I-\epsilon^2(\nabla^2 H_0\ J(x))^2\right)^{-1} =\sum_{i=0}^{n-1}r_i(x,\epsilon)A^i,
\end{equation}
where the functions $r_i(x,\epsilon)$ satisfy differential equations
\beq \label{grads r}
\nabla r_{i-1}(x,\epsilon) =A\nabla r_i(x,\epsilon), \quad i=1,\ldots,n-1.
\eeq
\end{lemma}
From \eqref{tH 5} and \eqref{inv}, we find:
\begin{eqnarray}
\t H_0(x,\epsilon) & = &  \sum_{i=0}^{n-1} r_i(x,\epsilon)x^{\rm T} A^i\nabla H_0(x)  \nonumber\\
& = & \sum_{i=0}^{n-1} r_i(x,\epsilon)x^{\rm T} \nabla H_i(x) \nonumber\\
& = & 2\sum_{i=0}^{n-1}r_i(x,\epsilon)H_i(x).  \label{tH thru H}
\end{eqnarray}
Differentiate formula \eqref{tH thru H}, taking into account  differential equations $A\nabla H_{i-1}=\nabla H_i$ and $A\nabla r_i=\nabla r_{i-1}$. We have:  
\begin{eqnarray} 
A\nabla \t H_0(x,\epsilon) & = & 2\sum_{i=0}^{n-1} \big(r_i(x,\epsilon)A\nabla H_i(x)+H_i(x)A\nabla r_i(x,\epsilon)\big) \nonumber \\
  & = & 2 \sum_{i=0}^{n-1} (r_i(x,\epsilon)\nabla H_{i+1}(x)+H_i(x)\nabla r_{i-1}(x,\epsilon))\nonumber \\
  & = & \nabla \Big( 2 \sum_{i=0}^{n-1} r_i(x,\epsilon)H_{i+1}(x)\Big).
\end{eqnarray}
By induction, we find:
\beq \label{Ak d tilde H}
A^m\nabla \t H_0(x,\epsilon)=\nabla  \Big( 2 \sum_{i=0}^{n-1} r_i(x,\epsilon)H_{i+m}(x)\Big), \quad m=0,1,\ldots, n-1.
\eeq
For any matrix polynomial
$
B=\beta_0I+\beta_1 A+\ldots+\beta_{n-1} A^{n-1},
$
the Hamilton function $K_0(x)$ of the corresponding vector field is defined by $\nabla K_0(x)=B\nabla H_0(x)$, and we have: $K_0(x)=\beta_0 H_0(x)+\beta_1 H_1(x)+\ldots+\beta_{n-1} H_{n-1}(x)$. As a consequence of the commutativity of $B$ and $A$, for the functions $K_i(x)$ defined by $\nabla K_i(x)=A^i\nabla K_0(x)$, we also have: $K_i(x)=\beta_0 H_i(x)+\beta_1 H_{i+1}(x)+\ldots+\beta_{n-1} H_{i+n-1}(x)$. Therefore, we derive from \eqref{Ak d tilde H}:
\beq \label{B d tilde H}
B\nabla \t H_0(x,\epsilon)=\nabla  \Big( 2 \sum_{i=0}^{n-1} r_i(x,\epsilon)K_{i}(x)\Big).
\eeq 
If now the vector fields with the Hamilton functions $H_0(x)$ and $K_0(x)$ are associated, that is, if $B^2=I$, then, according to \eqref{d2HJ2}, we have:
$$
 \left(I-\epsilon^2(\nabla^2 K_0\ J(x))^2\right)^{-1}=\left(I-\epsilon^2(\nabla^2 H_0\ J(x))^2\right)^{-1} =\sum_{i=0}^{n-1}r_i(x,\epsilon)A^i,
$$
so that
$$
\t K_0(x,\epsilon)=2\sum_{i=1}^{n-1} r_i(x,\epsilon)K_i(x).
$$
Comparing this with \eqref{B d tilde H}, we arrive at equation \eqref{grad tK}.
\end{proof}

\begin{cor}
The functions $\t H_0(x,\epsilon)$ and $\t K_0(x,\epsilon)$ are in involution with respect to both Poisson brackets, the original one with the Poisson tensor $J(x)$ and the perturbed one with the Poisson tensor $\Pi(x)$.
\end{cor}
\begin{proof}
The first statement follows directly from \eqref{grad tK}.  For the second statement, we compute, according to \eqref{Pi short} and \eqref{Jd2H as polynomial}:
\begin{eqnarray*}
\{\t H_0(x,\epsilon),\t K_0(x,\epsilon)\}_{\Pi} & = & (\nabla \t H_0)^{\rm T} \Pi(x) \nabla\t K _0\\
& = & (\nabla \t H_0)^{\rm T} \Big(I-\epsilon^2(J(x)\nabla^2 H_0)^2\Big)J(x)\nabla\t K _0\\
& = & (\nabla \t H_0)^{\rm T}\Big(I-\epsilon^2\sum_{i=0}^{n-1} q_i(x)A^i\Big)J(x)\nabla\t K _0\\ 
& = & (\nabla \t H_0)^{\rm T}\Big(I-\epsilon^2\sum_{i=0}^{n-1} q_i(x)A^iB^{\rm T}\Big)J(x)\nabla\t H _0=0,
\end{eqnarray*}
because the (diagonal blocks of the) matrices $A^iB^{\rm T}$ are circulant matrices.
\end{proof}

{\em Proof of Lemma \ref{lemma inv}.} We have:
$$
I-\epsilon^2(\nabla^2 H_0 J(x))^2=\big(I-\epsilon^2(J(x)\nabla^2 H_0)^2\big)^{\rm T}=I-\epsilon^2\sum_{i=0}^{n-1}q_i(x)A^{-i}.
$$
Since the inverse of a circulant matrix is also circulant, we have:
\begin{equation} \label{Pi inv}
\Big(I-\epsilon^2\sum_{i=1}^{n-1}q_i(x)A^{-i}\Big)^{-1}=\sum_{i=0}^{n-1}r_i(x,\epsilon)A^i.
\end{equation}
The coefficients $r_i(x)$ are determined just from the first column of the matrix identity obtained by left multiplying the right-hand side of the previous equation by the inverse of the left-hand side:
$$
\begin{pmatrix} 1-\epsilon^2 q_0 & -\epsilon^2 q_1 & -\epsilon^2 q_2 & \ldots & -\epsilon^2 q_{n-1}\\
                          -\epsilon^2 q_{n-1} & 1-\epsilon^2 q_0 & -\epsilon^3 q_1 & \ldots & -\epsilon^2 q_{n-2} \\
                          \ldots & \ldots & \ldots & \ldots & \ldots \\
                          -\epsilon^2 q_1 & -\epsilon^2 q_2 & -\epsilon^2 q_3 & \ldots & 1-\epsilon^2 q_0 
\end{pmatrix}  \begin{pmatrix} r_0 \\  r_1\\ \ldots \\ r_{n-1}  \end{pmatrix} =\begin{pmatrix} 1 \\ 0 \\ \ldots \\ 0 \end{pmatrix}.
$$
A straightforward check shows that the unique solution of this system is given by
\beq \label{ri}
r_k(x,\epsilon)=\frac{1}{n}\sum_{j=0}^{n-1} \frac{\omega^{kj}}{s_j(x,\epsilon)}, \quad k=0,\ldots,n-1,
\eeq
where
\beq \label{sj} \nonumber
s_j(x,\epsilon) =  1-\epsilon^2 \sum_{m=0}^{n-1} \omega^{jm}q_m(x), \quad \omega=\exp(2\pi i/n).
\eeq
It remains to show that functions \eqref{ri} satisfy differential equations \eqref{grads r}. We compute:
$$
\nabla r_k=\frac{\epsilon^2}{n}\sum_{j=0}^{n-1}\frac{\omega^{kj}}{s^2_j}\sum_{m=0}^{n-1}\omega^{jm}\nabla q_m
=\frac{\epsilon^2}{n}\sum_{j=0}^{n-1}\frac{1}{s^2_j}\left(\sum_{m=0}^{n-1}\omega^{j(k+m)}A^m\right)\!\nabla q_0.
$$
Now \eqref{grads r} follows from the obvious relation
$$
A\left(\sum_{m=0}^{n-1}\omega^{j(k+m)}A^m\right)=\sum_{m=0}^{n-1}\omega^{j(k+m-1)}A^m. \qquad\qquad  \qed
$$

\section{Examples}
\label{sect dim=4}

{\bf Dimension $2n=4$.}
Here, we are dealing with the following Lie-Poisson tensor:
$$
J(x)=\begin{pmatrix}
0 & 0 & x_1 &x_2 \\
0 & 0 & x_2 & x_1 \\
-x_1& -x_2 & 0 &0\\
-x_2 &-x_1 & 0 & 0
\end{pmatrix}.
$$
In coordinates, the non-vanishing Poisson brackets are
\begin{align*}
& \{x_1,x_3\}= \{x_2,x_4\}=x_1, \\
& \{x_1,x_4\}=  \{x_2,x_3\}=x_2
\end{align*}
(and those being obtained from these ones by skew-symmetry). 

General solution of equation \eqref{cond A} with $n=2$ is
\beq \label{A 4D}
A=\begin{pmatrix}
\alpha_0 & \alpha_1 & 0 &0 \\
\alpha_1 & \alpha_0 & 0 & 0 \\
0 & 0 & \alpha_0 & \alpha_1 \\
0 & 0 & \alpha_1 & \alpha_0
\end{pmatrix},
\eeq
and the corresponding set of functions in involution is generated by
\beq \label{A 4D red}
A=\begin{pmatrix}
0 & 1 & 0 & 0\\
1 & 0 & 0 & 0 \\
0 & 0 & 0 & 1\\
0 & 0 & 1 & 0
\end{pmatrix} .
\eeq
Since $A^2=I$, we can take $B=A$ for the associated vector field in this case. Note also that $A=A^{\rm T}$, a peculiarity of the case $2n=4$.

The Hesse matrices $\mathcal H$ of admissible Hamilton functions $H_0(x)=H(x)$ satisfying \eqref{d2H} are of the form
\beq
\mathcal H=\begin{pmatrix}  h_1 & h_2 & h_3 & h_4 \\
                                              h_2 & h_1 & h_4 & h_3 \\
                                              h_3 & h_4 & h_5 & h_6 \\
                                              h_4 & h_3 & h_6 & h_5 \end{pmatrix}.
\eeq
The Hesse matrix of the commuting Hamilton function $H_1(x)=K(x)$ is obtained by the simultaneous flips of the pairs $a=(h_1,h_2)$, $b=(h_3,h_4)$,
and $c=(h_5,h_6)$.

Lemma \ref{lemma matrix polynomial} in the present case of dimension $2n=4$ gives the following results: 
\beq
\big(J(x)\nabla^2 H\big)^2=q_0(x)I+q_1(x)A=\begin{pmatrix} q_0(x) & q_1(x) & 0 & 0 \\
                                                                    q_1(x) & q_0(x) & 0 & 0 \\
                                                                      0   &   0   & q_0(x) & q_1(x) \\
                                                                      0   &   0   & q_1(x) &  q_0(x) \end{pmatrix},
\eeq
where
\begin{eqnarray}
q_0(x)  & = & \frac{1}{4}{\rm tr}\big ((J(x)\nabla^2 H)^2\big)= \alpha (x_1^2+x_2^2)+2\beta x_1x_2,  \label{p 4D}\\
q_1(x) & = & \frac{1}{4}{\rm tr}\big(A(J(x)\nabla^2 H)^2\big)=\beta(x_1^2+x_2^2)+2\alpha x_1x_2, \label{q 4D}
\end{eqnarray}
with
$$
\alpha=h_3^2+h_4^2-h_1h_5-h_2h_6, \quad \beta= 2h_3h_4-h_1h_6-h_2h_5.
$$
Functions $q_0(x)$, $q_1(x)$ satisfy 
\beq \label{grads pq 4D}
\nabla q_1(x)=A \nabla q_0(x).
\eeq
This follows from Lemma \ref{lemma grads pq}, but is also obvious from the explicit formulas \eqref{p 4D}, \eqref{q 4D}.

\medskip

{\bf Dimension $2n=6$.} Here, we are dealing with the following Lie-Poisson tensor:
$$
J(x)=\begin{pmatrix}
0 & 0 & 0 & x_1 & x_2 & x_3 \\
0 & 0 & 0 & x_2 & x_3 & x_1 \\
0 & 0 & 0 & x_3 & x_1 & x_2  \\
-x_1 & -x_2 & -x_3 & 0 & 0 & 0 \\
-x_2 & -x_3 & -x_1 & 0 & 0 & 0 \\
-x_3 & -x_1 & -x_2 & 0 & 0 & 0
\end{pmatrix}.
$$
In coordinates, the non-vanishing Poisson brackets are
\begin{align*}
& \{x_1,x_4\}= \{x_2,x_6\}=\{x_3,x_5\}=x_1, \\
& \{x_1,x_5\}=  \{x_2,x_4\}=\{x_3,x_6\}=x_2,\\
& \{x_1,x_6\}=  \{x_2,x_5\}=\{x_3,x_4\}=x_3.
\end{align*}
General solution of \eqref{cond A} with $n=3$ is
\beq \label{A 6D}
A=\begin{pmatrix}
\alpha_0  & \alpha_2  & \alpha_1 & 0 & 0 & 0 \\
\alpha_1 & \alpha_0  & \alpha_2 & 0 & 0 & 0 \\
\alpha_2 & \alpha_1 & \alpha_0  & 0 & 0 & 0 \\
0 & 0 & 0 & \alpha_0  & \alpha_2 & \alpha_1 \\
0 & 0 & 0 & \alpha_1 & \alpha_0  & \alpha_2\\
0 & 0 & 0 & \alpha_2 & \alpha_1 & \alpha_0 
\end{pmatrix},
\eeq
and the corresponding set of functions in involution is generated by
\beq \label{A 6D red}
A=\begin{pmatrix}
0 & 1 & 0 & 0 & 0 & 0 \\
0 & 0 & 1 & 0 & 0 & 0\\
1 & 0 & 0 & 0 & 0 & 0\\
0 & 0 & 0 & 0 & 1 & 0 \\
0 & 0 & 0 & 0 & 0 & 1 \\
0 & 0 & 0 & 1 & 0 & 0
\end{pmatrix} .
\eeq
The Hesse matrices $\mathcal H$ of admissible Hamilton functions $H_0(x)$ satisfying \eqref{d2H} are of the form
\beq
\mathcal H=\begin{pmatrix}  h_1 & h_2 & h_3 & h_4 & h_5 & h_6\\
                                              h_2 & h_3 & h_1 & h_5 & h_6 & h_4 \\
                                              h_3 & h_1 & h_2 & h_6  & h_4 & h_5 \\
                                              h_4 & h_5 & h_6 & h_7 & h_8 & h_9 \\
                                              h_5 & h_6 & h_4 & h_8 & h_9 & h_7 \\
                                              h_6 & h_4 & h_5 & h_9 & h_7 & h_8 \end{pmatrix}.
\eeq
The Hesse matrices of the commuting Hamilton functions $H_1(x)$, $H_2(x)$ are obtained by the simultaneous cyclic shifts of the triples 
$a=(h_1,h_2,h_3)$, $b=(h_4,h_5,h_6)$, and  $c=(h_7,h_8,h_9)$. 

The associated vector fields are produced by three linearly independent matrices
\begin{equation}
B_j=\alpha_j I+\beta_j A+\gamma_j A^2=\cB_j\oplus\cB_j\quad (j=0,1,2),
\end{equation}
satisfying $B_j^2=I$ and $B_0+B_1+B_2=I$.
Their diagonal $3\times 3$ blocks are equal to:
\beq \label{B1 6D}
\cB_0=\frac{1}{3}\begin{pmatrix} 1 & -2 & -2 \\
                                  -2 & 1 & -2 \\
                                  -2 & -2 & 1  \end{pmatrix}, 
\eeq
\beq \label{B23 6D}                                 
\cB_1= \frac{1}{3}\begin{pmatrix} 1 & -2\omega & -2\omega^2 \\
                                  -2\omega^2 & 1 & -2\omega \\
                                  -2\omega & -2\omega^2 & 1  \end{pmatrix}, \quad                                  
\cB_2= \frac{1}{3}\begin{pmatrix} 1 & -2\omega^2 & -2\omega \\
                                  -2\omega & 1 & -2\omega^2 \\
                                  -2\omega^2 & -2\omega & 1  \end{pmatrix},                              
\eeq
where $\omega=\exp(2\pi i/3)$.

Lemma \ref{lemma matrix polynomial} in the present case of dimension $2n=6$ gives the following results: 
\begin{eqnarray}
\big(J(x)\nabla^2 H\big)^2 & = & q_0(x)I+q_1(x)A+q_2(x)A^2 \\
                                         & =  &
                                         \begin{pmatrix}   q_0(x) & q_1(x) & q_2(x) & 0 & 0 & 0 \\
                                                                    q_2(x) & q_0(x) & q_1(x) & 0 & 0 & 0 \\
                                                                    q_1(x) & q_2(x) & q_0(x) & 0 & 0 & 0 \\
                                                                        0     &     0     &     0     & q_0(x) & q_1(x) & q_2(x) \\
                                                                        0     &     0     &     0     & q_2(x) & q_0(x) &  q_1(x) \\
                                                                        0     &     0     &     0     & q_1(x) & q_2(x) &  q_0(x)
                                                                        \end{pmatrix}.
\end{eqnarray}
The functions $q_0(x)$, $q_1(x)$, $q_2(x)$ satisfy, according to Lemma \ref{lemma grads pq}, following relations:
\beq \label{grads q 6D}
\nabla q_1(x)= A\nabla q_0(x), \quad \nabla q_2(x)=A^2\nabla q_0(x). 
\eeq 
\section{Conclusions}

Completely integrable Hamiltonian systems lying at the basis of our constructions, seem to be worth studying on their own. In particular, their invariant $n$-dimensional varieties are intersections of $n$ hyperquadrics in the $2n$-dimensional space. It will be interesting to find out whether they are (affine parts) of Abelian varieties, that is, whether our systems are algebraically completely integrable. Still more interesting and intriguing are the algebraic-geometric aspects of the commuting systems of integrable maps introduced here. This will be the subject of our future research.

\section{Acknowledgements}

This research is supported by the DFG Collaborative Research Center TRR 109 ``Discretization in Geometry and Dynamics''.


\end{document}